\documentclass{lmcs}


\title{Coalgebraic Path Constraints}

\usepackage[utf8]{inputenc}
\usepackage{mathtools}

\usepackage{stmaryrd}
\usepackage{amssymb, amsmath}

\usepackage{graphicx}
\usepackage{wrapfig}
\usepackage[dvipsnames]{xcolor}
\usepackage[colorlinks, citecolor=blue, linkcolor=blue, urlcolor=purple]{hyperref}

\usepackage{tikz}
\usetikzlibrary{
    cd,
    fit,
    calc,
    positioning,
    arrows,
    automata,
    shapes
}
\tikzset{
    baseline = (current bounding box.center),
    every state/.append style = {
		inner sep = 3pt,
		minimum size = 18pt,
		initial text = {},
        fill=white
	},
	every edge/.append style = {
		->,
		>=stealth,
		bend angle=10,
		thick
	}
}
\usepackage[capitalise]{cleveref}

\Crefname{defi}{Definition}{Definitions}
\Crefname{thm}{Theorem}{Theorems}
\Crefname{exa}{Example}{Examples}
\Crefname{lem}{Lemma}{Lemmas}
\Crefname{cor}{Corollary}{Corollary}
\Crefname{prop}{Proposition}{Propositions}
\Crefname{rem}{Remark}{Remarks}

\newcommand{\Set}{\mathbf{Set}}
\newcommand{\VSp}{\mathbf{Vec}}

\newcommand{\Met}{\mathbf{Met}}

\newcommand{\Ord}{\mathbf{Ord}}
\newcommand{\Pos}{\mathbf{Pos}}
\newcommand{\Cat}{\mathbf{C}}
\newcommand{\VCat}{{\mathcal{V}\mathbf{Cat}}}

\newcommand{\VV}{\mathcal{V}}

\newcommand{\Coalg}{\mathsf{Coalg}}

\newcommand{\Cov}{\mathbf{V}}

\newcommand{\Shape}{\operatorname{Shape}}
\newcommand{\Id}{\mathrm{Id}}
\newcommand{\op}{\mathrm{op}}

\newcommand{\Sys}{\mathcal{S}}

\newcommand{\id}[1]{\mathsf{id}_{#1}}

\newcommand{\incl}{\mathsf{incl}}
\newcommand{\proj}{\mathsf{proj}}

\newcommand{\Pow}{\mathcal{P}}
\newcommand{\Dis}{\mathcal{D}}

\newcommand{\fin}{\mathsf{fin}}

\newcommand{\NN}{\mathbb{N}}
\newcommand{\ZZ}{\mathbb{Z}}
\newcommand{\RR}{\mathbb{R}}

\newcommand{\card}[1]{| #1 |}
\newcommand{\locCard}[1]{| #1 |_{\mathsf{loc}}}

\newcommand{\supp}{\mathrm{supp}}
\newcommand{\Grph}{\operatorname{Grph}}

\newcommand{\NatForm}{\mathsf{F}_{\mathrm{nat}}}

\newcommand{\der}{\partial}

\newcommand{\beh}{\mathsf{beh}}

\newcommand{\sem}[1]{\llbracket{#1}\rrbracket}

\newcommand{\e}{\varepsilon}
\newcommand{\bang}{\mathsf{!}}

\newcommand{\Word}[1]{\mathbf{Net}_{#1}}
\newcommand{\Net}{\mathcal{N}}
\newcommand{\WFunc}[1]{\sem{#1}}
\newcommand{\Incl}{\mathcal I}

\newcommand{\tr}[1]{
        \raisebox{-2pt}{
            \(\xrightarrow{\scriptstyle #1}\)
        }
}

\newcommand{\code}[1]{\mathtt{#1}}

\newcommand{\Bigmid}{~\Big|~}

\newcommand{\AND}{~\mathrm{and}~}

\begin{document}

\author[Todd Schmid]{Todd Schmid\lmcsorcid{0000-0002-9838-2363}}[a]
\address{Bucknell University, USA}
\email{t.schmid@bucknell.edu}

\begin{abstract}
    Axiomatizing covarieties of coalgebras for an endofunctor is less intuitive than axiomatizing varieties of algebras via equations (Dahlqvist and Schmid, 2022). 
    We introduce equational path constraints, a well-behaved and relatively easy to describe class of finitary behavioural properties that provide an algebra-flavoured alternative to coequations. The basic idea is to assign a pair of values to each path through a coalgebra and posit that the two values coincide. We show that equational path constraints define covarieties and construct final coalgebras relative to equational path constraints in some concrete cases. We connect equational path constraints to coequations when values computed from paths live in a monad, and we compute an upper bound on the number of colours needed to express the coequation. One of our constructions is reminiscent of the initial/terminal sequences of (Ad\'amek, 1974) and (Barr, 1993). Motivating examples include commutativity conditions in automata theory, differential equations, bi-infinite streams, and frame conditions.
\end{abstract}

\maketitle

\tableofcontents

\section{Introduction}
\label{sec:intro}

In the late 1980s/early 1990s, it was realized that many state-based systems from computer science could be encoded as coalgebras for endofunctors on a category. 
Formally, coalgebras are dual to algebras: if we think of the application of an endofunctor \(F\) to a set \(X\) as forming a set \(FX\) of compound objects (for e.g., from a family of abstract operations), then algebras model evaluation of compound elements to values \(FX \to X\), while coalgebras model the dynamics of state based systems as the production of compound elements from values \(X \to FX\).
Initial developments in the general theory of coalgebras, or \emph{universal coalgebra}, frequently took inspiration from the initial developments of universal algebra, especially the work of Birkhoff~\cite{Birkhoff35}.
In universal algebra, Birkhoff's variety theorem characterizes the classes of algebras closed under products, subalgebras, and quotients (called \emph{structural varieties}) as the classes of algebras defined by operations and equations (called \emph{varieties}). 
Many of the first papers on universal coalgebra were interested in providing a useful dual version of Birkhoff's theorem for coalgebras, with a particular focus on finding the right notion of \emph{coequation}~\cite{Rutten00,Rocsu1998birkhoff,Goldblatt2001,Gumm2000,Hughes2002,Jacobs1995} (see also~\cite{DahlqvistS21}).

A particularly influential dual version of Birkhoff's variety theorem is due to Rutten~\cite{Rutten00} (later improved by Gumm~\cite{Gumm2000}).
Rutten's covariety theorem characterizes \emph{structural covarieties}, the classes of coalgebras closed under subobjects, coproducts, and quotients, as the classes of coalgebras defined by their adherence to a specific set of system behaviours (called \emph{covarieties}).
Formally, these behaviours are elements of a cofree coalgebra over some generating set (elements of which are called \emph{colours}), so in this sense, a coequation is a predicate on a cofree coalgebra (see \cref{def:coequation}). 
Predicate coequations are the dual notion to equations if one is happy with the characterization of equations as quotients of a free algebra.

Rutten's covariety theorem is probably the simplest to use among the cited covariety theorems.
Specifying a predicate coequation amounts to describing a set of behaviours with some additional \emph{colouring} information (given by the set generating the cofree coalgebra).
But even still, relative to the number of tools for axiomatizing varieties of algebras, our technology for \emph{coaxiomatizing} covarieties of coalgebras is still very limited.


With the goal of adding to the coaxiomatization toolbox, we introduce \emph{equational path constraints}, a class of powerful but simple behaviour specifications that lend themselves well to intuitive equational descriptions.
Being models of state based systems, coalgebras support abstract notions of \emph{state}, \emph{transition}, and \emph{path}.
Informally, a \emph{path constraint} is a predicate on finite paths with the same starting state, and a coalgebra satisfies the path constraint if every path satisfies this predicate.
An \emph{equational} path constraint takes this a step further by defining this predicate via an equation between two ways of extracting a value of some type from a path leaving a given state.

Our intention in this paper is to take the first steps towards developing the basic theory of path constraints and equational path constraints.
Using equational path constraints, we are able to recapture many examples of covarieties drawn from the literature, as well as introduce some new ones.
In fact, the first two of our main results show that equational path constraints always define covarieties (\cref{thm:path covariety}), and that the converse is also true for finitary endofunctors on the category of sets and functions (\cref{thm:converse to covariety}), i.e., equational path constraints are equal in power to predicate coequations.

Beyond simply recapturing examples, we show that there are significant consequences for the structure of a covariety characterized by equational path constraints. 
Our third main result has to do with the \emph{chromatics} of covarieties of coalgebras in the category of sets and functions.
From the covariety theorem of Rutten~\cite{Rutten00} and the fact that equational path constraints coaxiomatize covarieties (\cref{thm:path covariety}), it follows that every equational path constraint is equivalent to some predicate coequation.
Unfortunately, neither theorem tells us much about this coequation; only that it exists.
One basic question we can ask about a coequation is how \emph{colourful} it is, i.e., how many colours are needed in its description~\cite{Adamek05}.
In \cref{sec:covariety}, we compute an upper bound on the number of colours needed to write down a coequation equivalent to an equational path constraint \emph{over a monad}. 
One instance of our bound shows that the commutativity and reversibility conditions on Moore automata (and many similar conditions) require at most 2 colours for their coequations.

The last of our main results is in \cref{sec:terminal net}. 
In that section, we construct a final \(F\)-coalgebra satisfying a given set of path constraints in an arbitrary complete category. 
For the construction to go through, we need \(F\) and each path constraint \(E\) to be sufficiently continuous and \emph{singular}, in a sense that we describe in \cref{sec:path constraints}.
We call the construction the \emph{terminal net construction}, as it is a kind of higher-dimensional version of the \emph{terminal sequence} construction of final coalgebras due to Ad\'amek~\cite{Adamek74} and Barr~\cite{Barr93}.






\paragraph{Relation to earlier version}
An earlier version of this paper was presented at the \emph{Coalgebraic Methods in Computer Science 2026} workshop (see~\cite{cmcsversionv2}).
In addition to the improved prose, several additions have been made to the present version: 
\cref{thm:converse to covariety,thm:bounded enriched branching}
are entirely new.
The discussion of frame conditions in \cref{subsec:frame conditions} has been significantly fleshed out.
The diagram in \cref{fig:singular net} has been added to illustrate the terminal net construction of \cref{thm:generalized Adamek-Barr construction}.
What we refer to as \emph{bounded branching} was alluded to in the workshop version, but not worked out in any detail.
Bounded branching is now the subject of~\cref{subsec:bounded branching}, where it has been given a fuller treatment.
Furthermore, all of the proofs that appeared in the appendix of the conference version now appear in the text.

\paragraph{Acknowledgements}
The paper review process is the backbone of the scientific community, and I sincerely appreciate every ounce of labour reviewers have put into this paper.
A quality review is not only thorough, but also helpful and encouraging, especially for an early-career researcher. 
I feel very lucky to have received so many quality reviews.
I am indebted to the reviewers at CSL'26 and CMCS'26, whose insight greatly improved the presentation of the paper (and in some cases, the results!).
I would also like to thank Fred Dahlqvist, Larry Moss, Alexandra Silva, Jurriaan Rot, Tobias Kapp\'e, and Henning Urbat for many helpful discussions and pointers to the literature.
The current version of the paper has also benefitted enormously from many productive conversations during the CMCS'26 workshop.


\section{Category Theory and Coalgebra Preliminaries}
\label{sec:coalgebra}
Familiarity with the basic language of category theory is assumed, but for the reader's convenience we set notation and reminders of the key definitions in this section.
Names of categories appear in bold: our primary examples will be
\begin{itemize}
    \item[] \(\Set\), the category of sets and functions;
    \item[] \(\VSp\), the category of (real) vector spaces and linear maps;
    \item[] \(\Met\), the category of metric spaces and nonexpansive maps; and
    \item[] \(\Pos\), the category of partially ordered sets and monotone maps.
\end{itemize}
\cref{thm:coax,thm:converse to covariety,thm:acc set eq,cor:set based terminal net} will only apply in the \(\Set\) case.

We use \(\Cat\) to denote an arbitrary category.  
We write \(\Id\) for the identity functor on \(\Cat\) that does nothing, defined by \(\Id X = X\) and \(\Id(f) = f\) for any object \(X\) and arrow \(f\colon X \to Y\).
For a given object \(B\) of \(\Cat\), we write \(\Delta_B\) for the \emph{constant} endofunctor on \(\Cat\) defined by \(\Delta_BX = B\) and \(\Delta_B(f) = \id{B}\) for any arrow \(f\).
Here, \(\id{B} \colon B \to B\)  is the identity.

A \emph{cone} for a diagram (formally, a functor) \(D \colon \mathbf D \to \Cat\) of \emph{shape} \(\mathbf D\) is a pair \((A, a)\) consisting of an object \(A\) of \(\Cat\) and a natural transformation \(a \colon \Delta_A \Rightarrow D\).
The diagram \(D\) is \emph{small} if the objects and arrows of \(\mathbf D\) form a set.
A \emph{cone homomorphism} \(f \colon (A, a) \to (B, b)\) for \(a \colon \Delta_A \Rightarrow D\) and \(b \colon \Delta_B \Rightarrow D\) is an arrow \(f \colon A \to B\) such that \(b_{j} \circ f = a_{j}\) for every object \(j\) of \(\mathbf D\).
A \emph{limit cone} for \(D\) is a final cone.
A category is complete if every small diagram has a limit cone.
Limit cones are unique up to cone isomorphism, so we may refer to a limit cone for a given diagram as \emph{the} limit of the diagram.
A functor \(F\) \emph{preserves} \(\mathbf D\)-limits if \((FA, F(a))\) is a limit cone for \(F \circ D\) whenever \((A, a)\) is a limit cone for \(D\).
A functor \(F\) \emph{reflects} \(\mathbf D\)-limits if \((A, a)\) is a limit cone whenever \((FA, F(a))\) is.\medskip
\begin{center}
    \framebox{\parbox{0.9\linewidth}{
    \textbf{\color{red}Assumption.} We assume throughout that \(\Cat\) is complete and that the functors \(F\) and \(H\) preserve monics. \
    Moreover, in concrete categories~\cite{AdamekHH09}, we will assume that \(F\) and \(H\) preserve inclusions, i.e., if \(U \subseteq X\), then \(FU \subseteq FX\). 
    }}
\end{center}
\medskip

\noindent%
The notions of \emph{cocone}, \emph{colimit}, and \emph{cocompleteness} are dual to those of cone, limit, and completeness, although we will not make any particular assumptions regarding colimits in general.
Standard references for category theory are~\cite{Riehl16,maclane1978categories}.

\paragraph{Coalgebra}
The paper is squarely situated in the subject of universal coalgebra.
Standard references are~\cite{Rutten00,Gumm99,Rutten01,Jacobs16}.

\begin{defi}
    [\(F\)-coalgebra]
	Given an endofunctor \(F\) on a category \(\Cat\), an \emph{(\(F\)-)coalgebra} is a pair \((X, \beta)\) consisting of an object \(X\) of \(\Cat\) and an arrow \(\beta \colon X \to FX\).
	The endofunctor \(F\) is called the \emph{coalgebraic signature} of \((X, \beta)\).
	A \emph{coalgebra homomorphism} between \(F\)-coalgebras \((X, \beta_X)\) and \((Y, \beta_Y)\) is an arrow \(h \colon X \to Y\) such that \(F(h) \circ \beta_X = \beta_Y \circ h\).
	We write \(h \colon (X, \beta_X) \to (Y, \beta_Y)\) to denote that \(h\) is a coalgebra homomorphism.
    The category of \(F\)-coalgebras and coalgebra homomorphisms is denoted \(\Coalg(F)\).
\end{defi}

Having concrete examples of coalgebras to keep in mind will be useful.
Given two endofunctors \(F,G \colon \Cat \to \Cat\), we write \(F \times G\) and \(F + G\) for the (pointwise) product and coproduct endofunctors, wherever defined, and for a set \(A\) we define \(\Id^A = \prod_{a \in A} \Id\).
In \(\Set\), \(\Id^A X = X^A\) coincides with the set of functions \(A \to X\).

\begin{exa}
    \label{eg:examples of coalgebras}
    Fix a set \(A\) and an arbitrary object \(B\), representing \emph{input letters} and an \emph{output space}, respectively.
    Note that we write \(\card X\) for the cardinality of a set.
    Also, for an arrow \(f \colon X \to Y \times Z\) we typically split \(f\) into a pairing of two maps \(f = \langle f_1, f_2\rangle\) defined so that composing with the first and second projections, \(\proj_1 \circ f = f_1\) and \(\proj_2 \circ f = f_2\).
    \begin{enumerate}
        \item \label{eg:it:B-colouring}
        A \emph{\(B\)-colouring} is a \(\Delta_B\)-coalgebra, i.e., the same data as an arrow \(X \to B\).
        
        \item \label{eg:it:transition} 
        A \emph{transition system} is a \(\Pow\)-coalgebra, where \(\Pow \colon \Set \to \Set\) is the covariant powerset functor. 
        For this functor, the action of \(\Pow\) on a function \(f \colon X \to Y\) takes the value of the image map \(\Pow(f) \colon \Pow X \to \Pow Y\) defined by \(\Pow(f)(U) = f(U) = \{f(x) \mid x \in U\}\) for each \(U \in \Pow X\).
        Given a transition system \((X, \beta)\), \(x, y \in X\), we write \(x \to y\) if \(y \in \beta(x)\).

        \item A \emph{finitely branching} transition system is a transition system where every state has at most finitely many outgoing transitions. 
        In other words, a finitely branching transition system is a \(\Pow_\fin\)-coalgebra, where \(\Pow_\fin X = \{U \subseteq X \mid \card{U} < \infty\}\).
        Similarly, for a cardinal \(\kappa\), a \emph{\(\kappa\)-branching} transition system is a \(\Pow_\kappa\)-coalgebra, where \(\Pow_\kappa X = \{U \subseteq X \mid \card{U} \le \kappa\}\).  
        
        \item \label{eg:it:LTS}
        A \emph{labelled transition system} or (\emph{LTS}) is like a transition system, but with added labels on transitions. 
        These are \(\Pow(\Delta_B \times \Id)\)-coalgebras, where \(B\) is a set of labels.
        Given a LTS \((X, \beta)\), \(x, y \in X\), we write \(x \tr{a} y\) if \((a, y) \in \beta(x)\).
        
        \item \label{eg:it:stream system} A \emph{stream system} (in any category) is a \(\Delta_B \times \Id\)-coalgebra \(\langle o, \partial\rangle \colon X \to B \times X\).
        
        \item 
        \label{eg:it:moore} 
        A \emph{Moore automaton} is a \(\Delta_B \times \Id^A\)-coalgebra.
        If \(\Cat\) is a \emph{concrete} category~\cite{AdamekHH09}, like \(\Set\), \(\VSp\), \(\Met\), and \(\Pos\), then given a Moore automaton \((X, \langle o, \delta\rangle)\), i.e., \(o \colon X \to B\) and \(\delta \colon X \to X^A\), we write \(x \tr{a} y\) if \(\delta(x)(a) = y\) and \(x \downarrow b\) if \(o(x) = b\).
        In \(\Set\), and with \(B = 2 = \{0,1\}\), this is the same data as a deterministic automaton in the classical sense.
        Moore automata in \(\VSp\) with \(B = \RR\) are \emph{linear weighted automata}~\cite{Bonchi2012}.    
        
        \item \label{eg:it:markov}
        For a function \(\theta \colon X \to [0,1]\), write \(\supp(\theta) = \{x \in X \mid \theta(x) > 0\}\). 
        A (\emph{finitely branching}) \emph{Markov chain} is a \(\Dis\)-coalgebra, where 
        \[
            \Dis X = \Big\{\theta \colon X \to [0,1] \Bigmid \card{\supp(\theta)} < \omega \AND \sum_{x \in X} \theta(x) = 1\Big\}
        \]
        is the set of finitely supported probability distributions on \(X\).
        The action of this functor on a function \(f \colon X \to Y\), is defined by 
        \[
            \Dis(f)(\theta)(y) = \sum_{f(x) = y} \theta(x)
        \]
        for any \(\theta \in \Dis Y\) and \(y \in Y\).
        We can interpret \(\Dis\) in \(\Set\) and \(\Met\), where in the \(\Met\) case we endow \(\Dis X\) with the so-called \emph{earth-mover's distance} given by the metric on \(X\)~\cite{FritzP2019}.%
        \footnote{ 
            It appears to be less well known that there is also a monad lifting of \(\Dis\) to \(\Pos\). 
            See, e.g.,~\cite{FRITZ2020107081}, or in the discrete case,~\cite{Schmid22}.
        }
        The functor \(\Dis\) is called the finitely supported probability distribution functor.
        Given a Markov chain \((X, \beta)\), we write \(x \tr{r} y\) if \(\beta(x)(y) = r\).

        \item \label{eg:it:lmc} 
        A \emph{Labelled Markov chain} (or \emph{LMC}) is like a Markov chain, but with added labels on transitions. 
        These coincide with \(\Dis(\Delta_B \times \Id)\)-coalgebras.
        
        \item \label{eg:it:compact} 
        A \emph{(nonexpansive) topological dynamical system}~\cite{Sternberg2010} is a nonexpansive endomap \(\beta \colon X \to X\) on a metric space \(X\).
        These coincide with \(\Id\)-coalgebras for \(\Id \colon \Met \to \Met\).
    \end{enumerate}
\end{exa}

\noindent%
For more examples, see~\cite{Rutten00,Gumm99,Jacobs16,Adamek05Intro,Rutten2019book}.

\paragraph{Final coalgebras}
For \(F \colon \Set \to \Set\), a state in an \(F\)-coalgebra can be thought of as a process that exhibits \emph{behaviour} as it unfolds, intuitively understood.
Coalgebra homomorphisms are the behaviour preserving maps between coalgebras, and two states \(x \in X\) and \(y \in Y\) of \(F\)-coalgebras \((X, \beta_X)\) and \((Y, \beta_Y)\) are called \emph{behaviourally equivalent} if there is a third \(F\)-coalgebra \((Z, \beta_Z)\) and coalgebra homomorphisms \(h \colon (X, \beta_X) \to (Z, \beta_Z)\) and \(k \colon (Y, \beta_Y) \to (Z, \beta_Z)\) such that \(h(x) = k(y)\).
Behavioural equivalence is an equivalence relation, and a behavioural equivalence class is known as a \emph{behaviour}. 
When the behaviours form a set (instead of a proper class), that set carries the structure of a \emph{final \(F\)-coalgebra}.

\begin{defi}
    [Finality]
    \label{def:final}
    In a full subcategory \(\Cov \subseteq \Coalg(F)\), an \(F\)-coalgebra \((Z, \zeta)\) is said to be \emph{final relative to \(\Cov\)} if \((Z, \zeta) \in \Cov\), and for every \((X, \beta) \in \Cov\) there is a unique coalgebra homomorphism \(\beh_\beta \colon (X, \beta) \to (Z, \zeta)\). 
    When \(\Cov = \Coalg(F)\), \((Z, \zeta)\) is just called a \emph{final coalgebra}. 
\end{defi}

Note that an \(F\)-coalgebra is final relative to \(\Cov\) if and only if it is a final object of the category \(\Cov\).
It was shown by Barr~\cite{Barr93} that if \(F \colon \Set \to \Set\) is \emph{accessible}, in the following sense, a final \(F\)-coalgebra exists.
A cardinal \(\kappa\) is said to be \emph{regular} if for any \(f \colon I \to \kappa\) with \(\bigcup_{i \in I} f(i) = \kappa\), \(\card I \ge \kappa\)~\cite{jech13}.

\begin{defi}
    [Accessible]
    \label{def:accessible}
    Let \(\kappa\) be a regular cardinal and \(F \colon \Set \to \Set\).
    Then \(F\) is \emph{\(\kappa\)-accessible} if for any set \(X\) and any \(t \in FX\), there exists \(U \subseteq X\) with \(\card U < \kappa\) such that \(t \in FU \subseteq FX\) (recall that \(F\) preserves inclusions).
    A functor is \emph{accessible} if it is \(\kappa\)-accessible for some regular cardinal \(\kappa\).
\end{defi}

Most endofunctors of interest in practical applications are accessible~\cite{Rutten00}.
Accessibility is a ``smallness'' condition that guarantees the behaviours form a set.
See also~\cite{KawaharaM00,Rutten00,Gumm2000,AdamekP01} for other conditions that deliver a final \(F\)-coalgebra.
%
In one of our main results (\cref{thm:coax}) a slightly different smallness condition is more convenient, as well as a related notion of cardinality for functors.

\begin{defi}
    [Local Cardinality]
    \label{def:small}
    Let \(\kappa\) be any cardinal and \(F \colon \Set \to \Set\).
    We say \(F\) is \emph{\(\kappa\)-small} if for any set \(X\) and any \(t \in FX\), there exists \(U \subseteq X\) with \(\card U \le \kappa\) such that \(t \in FU \subseteq FX\).
    A functor is \emph{small} if it is \(\kappa\)-small for some \(\kappa\).
    The \emph{local cardinality} \(\locCard{F}\) of a small \(F\) is the least \(\kappa\) such that \(F\) is \(\kappa\)-small.
\end{defi}

For example, \(\locCard{\Id} = 1\), and more generally, \(\locCard{\Id^A} = \card{A}\), because for any \(\theta \colon A \to X\), \(\card{\{\theta(x) \mid x \in X\}} \le \card{A}\).
The finite powerset functor \(\Pow_{\fin}\) and the countable powerset functor \(\Pow_\omega\) are \(\aleph_0\)-small, but the (unbounded) powerset functor \(\Pow\) is not small at all.
The difference between \(\kappa\)-small and \(\kappa\)-accessible is visible with \(\Pow_\omega\), which is \(\aleph_0\)-small but not \(\aleph_0\)-accessible (it is \(\aleph_1\)-accessible instead).
However, there is no largest regular cardinal, so a functor is small if and only if it is accessible.

\paragraph{The terminal sequence} 
\label{par:terminal}
It was mentioned above that every accessible endofunctor \(F \colon \Set \to \Set\) admits a final \(F\)-coalgebra.
In fact, one can explicitly construct a final \(F\)-coalgebra as the limit of a particular sequence of iterated approximations.
Writing \(1\) for a terminal object of \(\Cat\) (this exists because \(\Cat\) is complete), roughly speaking, the terminal sequence iteratively constructs a final \(F\)-coalgebra by repeatedly applying \(F\) to the terminal map \(\bang \colon F1 \to 1\) in \(\Cat\). 
This sequence is called the \emph{terminal sequence}.

\begin{defi}[Terminal Sequence]
    \label{def:terminal sequence}
    Let \(F \colon \Cat \to \Cat\). 
    Following~\cite{Barr93,Worrell98}, we transfinitely define the objects \(S_\lambda\) and arrows \(\bang^\lambda_\tau \colon S_\lambda \to S_\tau\) for ordinals \(\lambda\) and \(\tau\) as follows:
    \begin{itemize}
        \item \(S_0 = 1\) and \(\bang^1_0 = \bang \colon F1 \to 1\) 
        \item for any ordinal \(\lambda\), \(S_{\lambda+1} = FS_\lambda\)
        \item for any ordinal \(\lambda\), \(\bang^\lambda_\lambda = \id{S_{\lambda}}\)
        \item for any ordinals \(\tau,\lambda\) with \(\tau < \lambda\), \(\bang^{\lambda+1}_{\tau+1} = F(\bang^{\lambda}_\tau)\)
        \item for any limit ordinal \(\lambda\), \((S_\lambda, \bang^\lambda)\) is a limit cone of the diagram \(\{\bang^\tau_\rho \colon S_{\tau} \to S_{\rho} \mid \tau,\rho < \lambda\}\)
        \item for any limit ordinal \(\lambda\), \(\bang^{\lambda+1}_\lambda \colon (S_{\lambda+1}, F(\bang^{\lambda})) \to (S_\lambda, \bang^\omega)\) is the unique cone homomorphism
    \end{itemize}
    The \emph{terminal sequence} is then the diagram \(S \colon \Ord^{\op} \to \Cat\), where \(\Ord\) is the poset of ordinals viewed as a category, defined by \(S\lambda = S_\lambda\) on ordinals and by \(S(\lambda \ge \tau) = \bang^\lambda_\tau\) on arrows in \(\Ord^\op\).
\end{defi}

Given an ordinal \(\lambda\), we refer to the full subdiagram of \(S\) consisting of objects \(S_{\tau}\) with \(\tau < \lambda\) as the \emph{initial \(\lambda\)-segment} of the terminal sequence.
E.g., the \((\omega2+2)\)-segment is
\begin{equation}
    \label{eq:terminal sequence}
    \begin{tikzcd}[column sep=2em]
        1
        & F1 
            \ar[l, "{\bang}"']
        & F^2 1 
            \ar[l, "{F(\bang)}"']
        & \cdots 
            \ar[l, "{FF(\bang)}"']
        & S_\omega 
            \ar[l] 
            \ar[ll, out=-160, in=-20, "{\bang_2^\omega}"] 
            \ar[lll, out=-140, in=-40, "{\bang_1^\omega}"]
        & FS_\omega 
            \ar[l, "{\bang^{\omega+1}_\omega}"']
        & \cdots 
            \ar[l, "{F(\bang^{\omega+1}_\omega)}"']
        & S_{\omega2} 
            \ar[l]
            \ar[ll, out=-160, in=-20, "{\bang_{\omega+1}^{\omega2}}"] 
            \ar[lll, out=-140, in=-40, "{\bang_{\omega}^{\omega2}}"]
        & S_{\omega2 + 1} 
            \ar[l, "{\bang^{\omega2+1}_{\omega2}}"']
    \end{tikzcd}
\end{equation}
Note that \(\omega2 = \omega + \omega\).
This particular initial segment is of significant importance to us because of the following theorem and its role in \cref{thm:converse to covariety}.

\begin{thm}[Worrell~\cite{Worrell98}]
    \label{thm:worrell}
    Let \(F \colon \Set \to \Set\) be a finitary endofunctor. 
    Then, in the terminal sequence in~\eqref{eq:terminal sequence}, 
    \begin{enumerate}
        \item for any \(n \in \NN\), \(\bang^{\omega + n + 1}_{\omega + n}\) is injective,
        \item the function \(\bang^{\omega2+1}_{\omega2}\) is invertible, and \((S_{\omega2}, (\bang^{\omega2+1}_{\omega2})^{-1})\) is a final \(F\)-coalgebra.
    \end{enumerate}
\end{thm}

We will also make use of the \(\omega\)-segment of the terminal sequence later in \cref{sec:terminal net}, where we describe a generalization of (the \(\omega\)-segment of) the terminal sequence construction.


\paragraph{Covarieties and coequations}
When manipulating and reasoning about multiple coalgebras at a time, there are a few basic constructions used frequently.

\begin{defi}[Covariety]
    \label{def:covariety}
    Let \(F \colon \Cat \to \Cat\) be an endofunctor.
    \begin{itemize}
        \item Let \(I\) be a set and \((X_i, \beta_i)\) be an \(F\)-coalgebra for each \(i \in I\) such that the coproduct \(\coprod_{i \in I} X_i\) exists in \(\Cat\). 
        Let \(\incl_i \colon X_i \to \coprod_{i\in I} X_i\) be the inclusion arrow for each \(i \in I\) and write \([f_i]_{i\in I}\) for the arrow induced by the universal property of coproducts applied to a family of arrows \(f_i \colon X_i \to Y\). 
        We obtain the \emph{coproduct formula} \(\coprod_{i \in I} (X_i, \beta_i) \cong (\coprod_{i \in I} X_i, [F(\incl_i) \circ \beta_i]_{i \in I})\).%

        \item A \emph{regular epi} is an arrow \(q \colon X \twoheadrightarrow Q\) for which there exists a pair \(l,r \colon R \rightrightarrows X\) such that \(R \rightrightarrows X \xrightarrow{q} Q\) is a coequalizer diagram~\cite{AdamekHH09}.
        A \emph{homomorphic image} of \((X, \beta)\) is an \(F\)-coalgebra \((Q, \beta_Q)\) with a regular epi \(q \colon X \twoheadrightarrow Q\) (in \(\Cat\)) such that \(q \colon (X, \beta) \to (Q, \beta_Q)\) is a coalgebra homomorphism.
        
        \item A \emph{subcoalgebra} of \((X, \beta)\) is an \(F\)-coalgebra \((U, \beta_U)\) with a monic arrow \(\incl_U \colon U \hookrightarrow X\) that is a coalgebra homomorphism \(\incl_U \colon (U, \beta_U) \to (X, \beta)\). 
    \end{itemize}
    A \emph{covariety} is a full subcategory \(\Cov \subseteq \Coalg(F)\) that contains all coproducts, subcoalgebras, and homomorphic images of \(F\)-coalgebras in \(\Cov\).%
    \footnote{%
        \label{fn:regularity}%
        This definition of covariety differs slightly from~\cite{Hughes2002,Goldblatt05}, where the underlying monics are regular instead of the epis.
        The two definitions coincide in \(\Set\), \(\VSp\), and any topos~\cite[IV.1.2,IV.5.3]{maclane1992sheaves} or Abelian category~\cite{maclane1978categories}, where all monics/epis are regular.
    }
\end{defi}

Regular epis are the categorified version of quotient maps.
In \(\Set\), \(\VSp\), and many other concrete categories, every arrow \(f \colon X \to Y\) factors into a regular epi followed by a mono, \(f = \incl_U \circ q \colon X \twoheadrightarrow U \hookrightarrow Y\), and then \(U\) is called the \emph{image of \(X\) under \(f\)}.
Note that since \(F\) preserves monics, \(\beta_U\) and \(\beta_Q\) are uniquely determined by \(\incl_U\) and \(q\) in \cref{def:covariety}. 

\begin{defi}[Predicate Coequation]
    \label{def:coequation}
    Let \(F\colon \Cat \to \Cat\) and \(K\) be an object of \(\Cat\). 
    Let \((Z_K, \langle\ell, \zeta_K\rangle)\) be a final coalgebra for the functor \(\Delta_K \times F\).
    Then a \emph{coequation} is a subobject \(C \subseteq Z_K\). 
    An \(F\)-coalgebra \((X, \beta)\) is said to \emph{satisfy} \(C\) if for any \(k \colon X \to K\) (called a \emph{colouring}; see \cref{eg:examples of coalgebras}\eqref{eg:it:B-colouring}), the unique \(\Delta_K\times F\)-coalgebra homomorphism \(\beh_k \colon (X, \langle k, \beta\rangle) \to (Z_K, \langle \ell, \zeta_K\rangle)\) factors through the inclusion \(C \hookrightarrow Z_K\).
    \[\begin{tikzcd}[ampersand replacement = \&]
        X \ar[rr, "{\langle k, \beh_k\rangle}"] \ar[drr, dashed]
        \&\& K \times Z_K \\
        \&\& C \ar[u, "{\incl_C}"']
    \end{tikzcd}\]
    We write \(\Coalg(F; C)\) for the class of \(F\)-coalgebras satisfying \(C\), and say that \(C\) is a \emph{coequational coaxiomatization} of \(\Coalg(F; C)\).

    The object \(K\) is called the \emph{colour palette}, and the \(F\)-coalgebra \((Z_K, \zeta_K)\) obtained from \((Z_K, \langle \ell, \zeta_K\rangle)\) is the \emph{cofree \(F\)-coalgebra in \(K\) colours}.
\end{defi}

A cofree coalgebra in \(K\) colours is the same data as an \(F\)-coalgebra \((X, \beta)\) equipped with an arrow \(k \colon X \to K\). 
This map \(k\) can be thought of as ``colouring'' the states in \(X\), and the behaviour \(\beh_k(x) \in Z_K\) as the unfurling of the process \(x\) with state names replaced by the colours of each state recorded as it progresses. 
Elements of \(Z_K\) are often called \emph{behaviour patterns in \(K\) colours}. 

For certain \(F\)-coalgebras with \(F \colon \Set\to \Set\), covarieties are precisely the classes of \(F\)-coalgebras coaxiomatized by coequations.
This is the content of \cref{thm:rutten-gumm} below, due to Rutten~\cite{Rutten00} and later refined~\cite{Gumm2000} and generalized~\cite{AdamekP03}.

\begin{thm}[\cite{Rutten00,Gumm2000}]
    \label{thm:rutten-gumm}
    Let \(F \colon \Set \to \Set\) be accessible. 
    Then a class of \(F\)-coalgebras \(\Cov\) is a covariety if and only if there is a colour palette \(K\) such that \(\Cov = \Coalg(F;C)\) for some coequation \(C \subseteq Z_K\).
\end{thm}

\begin{defi}[Chromatic Number]
    \label{def:chromatic number}
	The \emph{chromatic number} of a covariety \(\Cov\) is the smallest cardinal \(\kappa\) (if it exists) such that there is colour palette \(K\) with \(\card{K} = \kappa\) and a coequation \(C \subseteq Z_K\) such that \(\Cov = \Coalg(F;C)\).
\end{defi}

The chromatic number of a covariety is related to the number of propositional variables needed to express a given frame condition as a modal formula (see \cref{rem:frame}): given a set of propositional variables \(\mathbb P\), a Kripke model is a coalgebra \(X \to \Pow\mathbb P \times \Pow X\). 
So, taking \(C = \Pow\mathbb P\), a frame condition in \(n\) propositional variables defines a coequation in \(2^n\) colours.
Chromatic numbers can also be thought of as dual to the numbers of variables needed to axiomatize varieties of algebras with equations. 
For example, the equational theory of monoids requires at most three variables to axiomatize.
Closely reading the proofs of~\cite[Thm.~17.5]{Rutten00} and~\cite[Thm.~4.1]{AdamekP01} reveals that for \(\kappa\)-accessible \(F\), we can take \(K = \kappa^+\) (the least cardinal greater than \(\kappa\)) in \cref{thm:rutten-gumm}.

\section{Path Constraints}
\label{sec:path constraints}

Typically, one models systems as coalgebras by finding a functor \(F\) for which the category of systems is precisely \(\Coalg(F)\).
But there are cases where the full category of \(F\)-coalgebras is too wide in scope. 
The example that initially motivated the current work comes from a paper of Boreale~\cite{Boreale19}, where partial differential equations find a coalgebraic model in terms of Moore automata (\cref{eg:examples of coalgebras}\eqref{eg:it:moore}).
Boreale's coalgebraic model of PDEs specifically requires \emph{commutativity}, a property of Moore automata that can be described locally as a constraint on paths of length \(2\).
See~\cref{eg:commutative Moore automata} for details.
In this section, we generalize the commutativity conditions of Boreale to a broader class of properties of \(F\)-coalgebras called \emph{path constraints}.
We develop some of the basic theory of path constraints and discuss a number of examples.

To start with, we need an abstract notion of \emph{path through an \(F\)-coalgebra}.
This is intuitive in the examples of \(F\)-coalgebras we saw in the previous subsection (automata, labelled transition systems, etc.), but for arbitrary coalgebras the right notion of path is more subtle.
Here, we treat a path through an \(F\)-coalgebra as a series of \emph{steps} obtained by repeatedly unfolding the coalgebra.

\begin{defi}
    \label{def:steps}
    For a natural number \(n \in \NN\), \(F^n\) is defined recursively by \(F^{0} = \Id\) and \(F^{n+1} = F \circ F^n\).
    Given an \(F\)-coalgebra \((X, \beta)\), we write \(\beta^n\colon X \to F^nX\) for the arrow recursively defined by \(\beta^0 = \id{X}\) and \(\beta^{n + 1} = F(\beta^n) \circ \beta\) for any \(n \in \NN\).
    We refer to \((X, \beta^n)\) as the \emph{\(n\)-step coalgebra} of \((X, \beta)\).
\end{defi}

Intuitively, \(\beta^n(x)\) records all of the paths of length \(n\) exiting the state \(x \in X\). 
The idea behind path constraints is that they are predicates on the possible paths of \(F\)-coalgebras that are preserved by the action of \(F\) on functions.
Soon we will see a broad definition that allows for relations \emph{between} paths instead of just predicates on paths of a single length, but we begin at the predicate level for the sake of scaffolding.

Let us say that a \emph{subfunctor} of \(F\) is a functor \(E\) with a natural transformation \(e \colon E \Rightarrow F\) such that \(e_X\) is monic for every object \(X\) of \(\Cat\).
We write \(e \colon E \subseteq F\) to express that \(E\) is a subfunctor of \(F\) equipped with \(e\).

\begin{defi}
    \label{def:singular path constraint} 
    Given \(n \in \NN\), a \emph{singular path constraint} \emph{(of length \(n\))} is a subfunctor \(e \colon E \subseteq F^n\).
    An \(F\)-coalgebra \((X, \beta)\) \emph{satisfies} \(E\) 
    if the \(n\)-step coalgebra of \((X, \beta)\) factors through \(e\), i.e., there exists \(\e_\beta \colon X \to EX\) such that the diagram
    \begin{equation}
        \label{eq:satisfy singular path constraint}
        \begin{tikzcd}[row sep = 0.5em]
            X 
            	\ar[r, "\beta"] 
            	\ar[drr, dashed, to path={ 
            		-- ([yshift=-3ex]\tikztostart.south) -> node[below, pos=0.5] {\(\e_\beta\)} (\tikztotarget)
            	}, rounded corners = 5pt]
            	\ar[rrrr, to path={ -- ([yshift=2ex]\tikztostart.north) -| node[above, pos=0.25] {\(\beta^n\)} (\tikztotarget)}, rounded corners=5pt]
            & FX \ar[r, "F(\beta)"]
            & F^2X \ar[r, "F^2(\beta)"]
            & \cdots \ar[r]
            & F^n X\\
            && EX \ar[urr, hook, to path={ 
            		-- ([xshift=20ex]\tikztostart.east)  node[below, pos=0.5] {\(e_X\)} -> (\tikztotarget)
            	}, rounded corners = 5pt]
        \end{tikzcd}
        \qquad
        e_X \circ \e_\beta = \beta^n
    \end{equation}
    commutes.
    We write \(\Coalg(F, E)\) for the full subcategory of \(\Coalg(F)\) consisting of \(F\)-coalgebras that satisfy \(E\).
\end{defi}

Intuitively, singular path constraints dictate how paths of a single fixed length can behave. 
Note that since \(e_X\) is monic, the factorization in~\eqref{eq:satisfy singular path constraint} is unique.

\begin{exa}
    \label{eg:trivial examples}
    Before getting into the interesting examples, it is always worth mentioning some trivial ones. 
    For an initial object \(\emptyset\) of \(\Cat\), one always has the empty path constraint \(\Delta_\emptyset \subseteq F^0 = \Id\) of length \(0\) (or of any length, really).
    Only the empty \(F\)-coalgebra \((\emptyset, ?)\) with \(? \colon \emptyset \to F\emptyset\) satisfies \(\Delta_\emptyset\). 
    There is also the trivial path constraint \(\id{} \colon E = F^0\) of length \(0\) (again, any length will do).
    Every \(F\)-coalgebra satisfies the trivial path constraint. 
\end{exa}

\begin{exa}
    \label{eg:commutative Moore automata}
    Commutativity of two inputs to a Moore automaton (\cref{eg:examples of coalgebras}\eqref{eg:it:moore}) is a singular path constraint of length \(2\).
    Let \(F = \Delta_B \times \Id^A\).
    Given an input letter \(a \in A\), define the \emph{\(a\)-derivative} \(\der a \colon F \Rightarrow \Id\) by \((\der a)_X(p, \theta) = \theta(a)\) for any set \(X\) and \((p, \theta) \in FX\).
    Note that \((\der a)_X\) is natural in \(X\).
    Given \(a,b \in A\), we define the singular path constraint
    \[
        EX = \big\{(p, \Theta) \in F^2X \mid (\der a)_{X} \circ (\der b)_{FX}(p, \Theta) = (\der b)_{X} \circ (\der a)_{FX} (p, \Theta)\big\}
    \]
    For brevity, we tend to drop subscripts and \(\circ\) between derivatives.
    Observe that \(E\) is a subfunctor of \(F^2\), and that for a Moore automaton \((X, \beta)\), \(\beta^2\) factors through \(EX\) if and only if \(\der a \der b\beta^2 = \der b \der a \beta^2\).
    Writing \(\beta = \langle o, \delta\rangle\), satisfaction of \(E\) is equivalent to 
    \begin{gather*}
        \delta(\delta(x)(b))(a) = \delta(\delta(x)(a))(b)
        \qquad 
        \begin{tikzpicture}[yscale=0.6]
           \node[state] (x) at (0, 0) {\(x\)};
           \node[state, dotted] (u) at (1.5, 0.5) {\phantom{\(u\)}};
           \node[state, dotted] (v) at (1.5, -0.5) {\phantom{\(v\)}};
           \node[state] (y) at (3, 0) {\phantom{\(y\)}};
           \draw (x) edge[bend left] node[above] {\(a\)} (u);
           \draw (x) edge[bend right] node[below] {\(b\)} (v);
           \draw (u) edge[bend left] node[above] {\(b\)} (y);
           \draw (v) edge[bend right] node[below] {\(a\)} (y);
       \end{tikzpicture}
    \end{gather*}
    for all \(x \in X\).
    That is, the actions of inputting \(a\) and \(b\) \emph{commute} in \((X, \beta)\).
\end{exa}

\begin{exa}
    \label{eg:determinism}
    Given a LTS \((X, \beta)\) (see \cref{eg:examples of coalgebras}\eqref{eg:it:LTS}), we say that it is \emph{deterministic} if \(\card{\{ y \mid x \tr{a} y\}} = 1\) for all \(x \in X\) and \(a \in A\).
    This is a path constraint of length \(1\).
    Given a function \(f \colon A \to B\), write \(\Grph(f) = \{(a, f(a)) \mid a \in A\}\).
    Then the subfunctor \(EX = \{\Grph(f) \mid f \colon A \to X\} \subseteq \Pow(A \times X)\) is the path constraint for determinism of LTSs.
    Clearly, \(\Coalg(\Pow(A \times \Id), E) \cong \Coalg(\Id^A)\).
\end{exa}

\subsection{Nonsingular path constraints}
There are many familiar examples of path-based properties that involve paths of multiple lengths.
For instance, \emph{transitivity} of a transition system (\cref{eg:transitivity} below) mentions paths of length \(2\) and of length \(1\): ``if there is a \emph{two}-step path \(x \to y \to z\), then there is a \emph{one}-step path \(x \to z\)''.
Transitivity can be captured by the generalization of singular path constraints below, which includes constraints on relations between paths of varying lengths.

\begin{defi}
    \label{def:path constraint}
    The set \(\Shape(F)\) of \emph{path shapes over \(F\)} is the collection of functors built from the formation rules below (where \(I\) ranges over any set),
    \[
        J, J_i, K ::= \Id \mid F \mid \textstyle\prod_{i \in I} J_i \mid K \circ J
    \]
    Given an \(F\)-coalgebra \((X, \beta)\) and \(J \in \Shape(F)\), the \emph{\(J\)-path coalgebra} \((X, \beta^{J})\) is defined recursively by 
    \(\beta^\Id = \id{X}\), 
    \(\beta^{F} = \beta\), 
    \(\beta^{\prod J_i} = \langle \beta^{J_i}\rangle_{i \in I} \colon X \to \prod_{i \in I} J_iX\), and 
    \(\beta^{K\circ J} = K(\beta^J) \circ \beta^K \colon X \to KX \to KJX\).
    A \emph{\(J\)-path constraint} (or \emph{\(J\)-constraint}) is a subfunctor \(e \colon E \subseteq J\). 
    An \(F\)-coalgebra \((X, \beta)\) \emph{satisfies} \(E\) if \(\beta^{J}\) factors through \(e\).
    We write \(\Coalg(F, E)\) for the full subcategory of \(F\)-coalgebras that satisfy \(E\).
\end{defi}

\begin{exa}
    \label{eg:transitivity}
    A transition system \((X, \beta)\) (\cref{eg:examples of coalgebras}\eqref{eg:it:transition}) is \emph{transitive} if \(x \to y \to z\) implies \(x \to z\).
    Transitivity is a \((\Pow \times \Pow^2)\)-constraint. 
    To see this, let \(EX = \{(U, \Phi) \mid \bigcup \Phi \subseteq U\} \subseteq \Pow X \times \Pow^2 X\).
    In a transition system \((X, \beta)\) with \(x \in X\), we have 
    \(\beta^2(x) = \{V \subseteq \Pow X \mid \exists y \in X,~ x \to y \text{ and } V = \beta(y)\}\).
    Then we have \(\bigcup \beta^2(x) = \{z \mid \exists y \in X,~ x \to y \to z\}\), 
    so \((\beta(x), \beta^2(x)) \in EX\) if and only if
    \[
        \{z \mid \exists y \in X,~ x \to y \to z\} 
        = \bigcup \beta^2(x) 
        \subseteq \beta(x) 
        = \{z \mid x \to z\}
    \]
    for all \(x \in X\).
    This is equivalent to \((X, \beta)\) being transitive.
\end{exa}

We will see more examples of nonsingular path constraints shortly. 
For now, let us record the following technical lemma that will be used throughout the paper. 

\begin{restatable}{lem}{jpathfunctorial}
    \label{lem:always a homom}
    For any \(J \in \Shape(F)\), the mapping defined by 
    \((X, \beta)^J = (X, \beta^J)\) for any \(F\)-coalgebra \((X, \beta)\) 
    and 
    \((h)^J = h\) for any coalgebra homomorphism \(h \colon (X, \beta) \to (Y, \gamma)\)
    is a functor \((-)^J \colon \Coalg(F) \to \Coalg(J)\) that preserves and reflects colimits, homomorphic images, and the subcoalgebra relation.
\end{restatable}

\begin{proof}
    We first need to show that if \(h \colon (X, \beta) \to (Y, \gamma)\) is an \(F\)-coalgebra homomorphism, then \(h^J \colon (X, \beta^J) \to (Y, \gamma^J)\) is a coalgebra homomorphism. 
    This proceeds by induction on \(J\). 
    \begin{description}
        \item[Base case.] In the base case, either \(J = \Id\) or \(J = F\). 
        The statement is trivial if \(J = \Id\), and the statement follows by assumption if \(J = F\). 

        \item[Induction step 1.] 
        Given a set \(I\), suppose that \(J_i \in \Shape(F)\) and that \(h\) is a \(J_i\)-coalgebra homomorphism for all \(i \in I\).
        Let \(J = \prod_{i \in I} J_i\).
        Then
        \begin{align*}
            \gamma^J \circ h^J
            &= \langle \gamma^{J_i}\rangle_{i \in I} \circ h  \tag{def.~\(\gamma^{\prod J_i}\)} \\
            &= \langle \gamma^{J_i} \circ h\rangle_{i \in I} \tag{universal prop.~of \(\prod\)}\\
            &= \langle J_i(h) \circ \beta^{J_i}\rangle_{i\in I} \tag{ind.~hyp.}\\
            &= (\prod_{i \in I} J_i)(h) \circ \langle \beta^{J_i}\rangle_{i \in I} \tag{universal prop.~of \(\prod\)}\\
            &= J(h^J) \circ \beta^{J} \tag{def.~\(\beta^{\prod J_i}\)}
        \end{align*}
        
        \item[Induction step 2.] 
        Now suppose that \(h\) is both a \(J_1\)-coalgebra homomorphism and a \(J_2\)-coalgebra homomorphism.
        Let \(J = J_2 \circ J_1\).
        Then 
        \begin{align*}
            \gamma^J \circ h^J
            &= J_2(\gamma^{J_1}) \circ \gamma^{J_2} \circ h  \tag{def.~\(\gamma^{J_2 \circ J_1}\)} \\
            &= J_2(\gamma^{J_1}) \circ J_2(h) \circ \beta^{J_2} \tag{\(h\) a \(J_2\)-coalg.~homom.}\\
            &= J_2(\gamma^{J_1} \circ h) \circ \beta^{J_2} \tag{\(J_2\) a functor}\\
            &= J_2(J_1(h) \circ \beta^{J_1}) \circ \beta^{J_2} \tag{ind.~hyp.}\\
            &= J_2J_1(h) \circ J_2(\beta^{J_1}) \circ \beta^{J_2} \tag{\({J_2}\) a functor}\\
            &= J(h) \circ \beta^{J} \tag{def.~\(\beta^{{J_2} \circ J_1}\)}
        \end{align*}
    \end{description}

    \noindent
    The key observations in the rest of this proof are two-fold: the first is that for any endofunctor \(G \colon \Cat \to \Cat\), the forgetful functor \(U^G \colon \Coalg(G) \to \Cat\) preserves and reflects (a.k.a., creates) colimits~\cite{Rutten00}.
    This then applies to both \(G = F\) and \(G = J\). 
    The second observation is that the triangle below commutes: 
    \[\begin{tikzcd}
        \Coalg(F) \ar[rr, "(-)^J"] \ar[dr, "U^F"'] && \Coalg(J) \ar[dl, "U^J"] \\
        & \Cat &
    \end{tikzcd}\]
    It can immediately be seen from the triangle that \((-)^J\) preserves and reflects the subcoalgebra relation and homomorphic images. 

    To see that \((-)^J\) preserves and reflects colimits, let \(D \colon \mathbf D \to \Coalg(F)\) be any small diagram and let 
    \(((X, \beta), c)\) 
    be a cocone for \(D\).
    Then \(c\) is a colimiting cocone if and only if 
    \((U^F(X, \beta), U^F(c)) = (X, U^F(c))\) 
    is a colimiting cocone for \(U^F \circ D\), because \(U^F\) preserves and reflects colimits~\cite[Theorem 4.5]{Rutten00}.
    But \(U^F = U^J \circ (-)^J\), so \((X, c)\) is a colimiting cocone for \(D\) if and only if 
    \((U^J \circ (-)^J (X, \beta), U^J(c^J)) = (X, U^J(c^J))\) 
    is a colimiting cocone for 
    \(U^F \circ D = U^J \circ (-)^J \circ D\).
    Since \(U^J\) preserves and reflects colimits (also~\cite[Theorem 4.5]{Rutten00}), 
    \((X, U^J(c^J))\) 
    is a colimiting cocone if and only if 
    \(((X, \beta^J), c^J)\) 
    is a colimiting cocone for \((-)^J \circ D\). 
    Chaining our reasoning together, \(c\) is a colimiting cocone if and only if \(c^J\) is a colimiting cocone.
    This shows that \((-)^J\) preserves and reflects colimits.
    In particular, \((-)^J\) preserves and reflects coproducts. 
\end{proof}

\subsection{Equational path constraints}
It was mentioned in the introduction that path constraints can often be described as equations.
Informally, for \(J \in \Shape(F)\) and \(e \colon E \subseteq J\), what this means is that the elements of \(EX\) are precisely the \(t \in JX\) that satisfy some equation \(\lambda(t) = \rho(t)\), where \(\lambda\) and \(\rho\) are two different ways to extract some value from \(t\).
We make this more precise in \cref{def:equational path constraint} below, which captures \(\lambda\) and \(\rho\) abstractly as natural transformations. 

\begin{defi}
    \label{def:equational path constraint}
    Given \(J \in \Shape(F)\), a \(J\)-path constraint \(e \colon E \subseteq J\) is said to be \emph{equational over \(H\)} if there are two natural transformations \(\lambda,\rho \colon J \Rightarrow H\), called \emph{left} and \emph{right} respectively, such that \(e_X\) is an equalizer of \(\lambda_X,\rho_X \colon JX \rightrightarrows HX\) for all \(X\).
    We say that \((X, \beta)\) \emph{satisfies \(\lambda \equiv \rho\)} if it satisfies the \(J\)-constraint \(E\) given by the left and right transformations \(\lambda,\rho \colon J \Rightarrow H\).
\end{defi}

Note that a path constraint can happen to be equational over multiple functors \(H\) (in fact, this is often the case).
From \cref{def:equational path constraint}, we immediately obtain the following lemma.

\begin{lem}
    \label{lem:equational}
    Let \(F\) be an endofunctor on \(\Cat\), \(\lambda,\rho \colon J \Rightarrow H\).
    Then an \(F\)-coalgebra \((X, \beta)\) satisfies \(\lambda \equiv \rho\) if and only if \(\lambda_X \circ \beta^J = \rho_X \circ \beta^J\).
\end{lem}

Commutativity of Moore automata (\cref{eg:commutative Moore automata}) can also be written as an equational path constraint over \(\Id\), using the left and right transformations \(\lambda, \rho \colon B \times (B \times \Id^A)^A \Rightarrow \Id\) defined by \(\lambda = \der a\der b\) and \(\rho = \der b \der a\).
Of course, this only covers two commuting letters.
To impose multiple path constraints simultaneously, we consider \emph{systems} of path constraints.

\begin{defi}
    \label{def:systems of pcs}
    A \emph{system of path constraints} is a set of path constraints on \(F\)-coalgebras.
    An \(F\)-coalgebra \emph{satisfies} a system \(\Sys\) of path constraints if it satisfies every path constraint in \(\Sys\).
    We write \(\Coalg(F, \Sys)\) for the full subcategory of \(F\)-coalgebras that satisfy \(\Sys\).
    A \emph{system of equational path constraints (over a functor \(H\))} is a system of path constraints that are equational (over \(H\)). 
\end{defi}

\begin{exa}
    \label{eg:independence}
    Conditions similar to commutativity can also be found in concurrency theory~\cite{Mazurkiewicz}. 
    Given an alphabet \(A\) of \emph{action symbols}, an \emph{independency relation} is a symmetric relation \(I \subseteq A^2\) that indicates which sequences of actions can be taken in any order without affecting the outcome.
    In a LTS \((X, \beta)\) (here, \(F = \Pow(A \times \Id)\)), if \((a, b) \in I\) and \(u,w \in A^*\), then every state \(x \in X\) executes the trace \(uabw\) if and only if it executes the trace \(ubaw\).
    This can be enforced using an equational path constraint involving the family of transformations \(\pi^{w} \colon F^{\mathsf{len}(w)}\Rightarrow \Pow\) defined by 
    \(\pi_X^\e(x) = \{x\}\), 
    \(\pi_X^a(U) = \{x \in X \mid (a,x) \in U\}\), and 
    \(\pi_X^{aw}(\Phi) = \bigcup \{\pi_{F^{\mathsf{len}(w)}}^{a}(U) \mid (a, U) \in \Phi\}\). 
    Then \(\pi^{ab} \equiv \pi^{ba}\) requires 
    \[
        \{z \mid \exists y,~ x \tr{a} y \tr{b} z\} 
        = \pi_X^{ab} \circ \beta^2(x) 
        = \pi_X^{ba} \circ \beta^2(x) 
        = \{z \mid \exists y,~x \tr{b} y \tr{a} z\}
    \]
    for every state \(x \in X\).
    The entire independence relation can then be modelled  as the system of path constraints \(\Sys_I = \{\pi^{ab} \equiv \pi^{ba} \mid (a,b) \in I\}\) over \(\Pow\).
\end{exa}

Interestingly, for an accessible endofunctor \(F\) on \(\Set\), every path constraint is equational.

\begin{defi}
    For a given cardinal \(\kappa\) and a given set \(X\), call a set of functions \(S \subseteq \{f \mid \exists \lambda < \kappa,~f \colon X \to \lambda\}\) a \emph{\(\kappa\)-co-sieve on \(X\)} if for any \(f \colon X \to \lambda_1\) in \(S\) and any \(g \colon \lambda_1 \to \lambda_2\) with \(\lambda_2 < \kappa\), \(g \circ f \in S\).
    We define the endofunctor \(\Omega_\kappa \colon \Set \to \Set\) by
    \[
        \Omega_\kappa X = \{S \mid \text{\(S\) is a \(\kappa\)-co-sieve on \(X\)}\}
        \qquad
        \Omega_\kappa(h)(S) = \{g \mid g \circ h \in S\}
    \] 
    for any set \(X\) and \(h \colon X \to Y\).
\end{defi}

The general idea behind the following theorem is that \(\Omega_\kappa\) mimics the canonical subobject classifier of \(\Set^{\Set_\kappa^\op}\)~\cite{maclane1992sheaves}, where \(\Set_\kappa\) is the category of sets of cardinality at most \(\kappa\).

\begin{restatable}{thm}{weaklyaccessibleequational}
    \label{thm:acc set eq}
    Let \(F\) be an accessible endofunctor on \(\Set\), \(J \in \Shape(F)\), and \(e \colon E \subseteq J\).
    Then \(E\) is equational over \(\Omega_\kappa\) for some regular cardinal \(\kappa\).
\end{restatable}

\begin{proof}
    Let \(F\) be an accessible endofunctor on \(\Set\) that preserves inclusions, and let \(e \colon E \subseteq J\) be any path constraint.
    Since \(J\) is formed by taking small (set-sized) products and finitely many compositions of functors formed from \(F\) and \(\Id\), both of which are accessible and preserve inclusions, \(J\) is \(\kappa\)-accessible for some regular cardinal \(\kappa\) and preserves inclusions.

    Now consider the following two natural transformations:
    first, the ``true'' transformation \(\top \colon J \Rightarrow \Omega_\kappa\), defined by
    \[
        \top_X(v) = \{f \mid \text{\(f \colon X \to \lambda\) for some \(\lambda < \kappa\)} \}
    \]
    and second, the transformation \(\chi^E \colon J \Rightarrow \Omega_\kappa\) defined by
    \[
        \chi_X^E (v) = \{f \mid \text{\(f \colon X \to \lambda\) with \(\lambda < \kappa\) and \(J(f)(v) \in E\lambda\)}\}
    \] 
    Given \(v \in EX \subseteq JX\) and \(f \colon X \to \lambda\) for \(\lambda < \kappa\), we know that \(E(f)(v) \in E\lambda\), by definition, and therefore \(J(f)(v) = E(f)(v) \in E\lambda\).
    This puts \(f \in \chi_X^E(v)\), so \(\chi^E_X(v) = \top_X(v)\). 
    
    Conversely, suppose \(v \in JX\) and \(\chi_X^E(v) = \top_X(v)\).
    This tells us that for any function \(f \colon X \to \lambda\) with \(\lambda <\kappa\), \(J(f)(v) \in E\lambda\).
    Since \(J\) is \(\kappa\)-accessible, there is a subset \(V \subseteq X\) such that \(v \in FV \subseteq FX\) and \(\card{V} < \kappa\).
    Find the cardinal \(\lambda <\kappa\) with a bijection \(j \colon V \cong \lambda\).
    Now, (using the axiom of choice) there is a retraction \(r \colon X \to V\) of the inclusion \(V \subseteq X\).
    We therefore know that \(v = J(r)(v)\), because \(J\) preserves inclusions and by functoriality \(J(r)\) is a retraction of the inclusion \(JV \subseteq JX\).
    We have 
    \[
        J(j)(v) 
        = J(j) \circ J(r)(v) 
        = J(j \circ r)(v) 
        \in J\lambda
    \]
    Since \(j \circ r \colon X \to \lambda < \kappa\), we have assumed that \(J(j \circ r)(v) \in E\lambda\), so \(J(j)(v) \in E\lambda\).
    Therefore, \(v = E(j^{-1}) \circ E(j)(v) \in EV \subseteq EX\), so \(v \in EX\).
    
    We have just shown that 
    \[
        EX = \{v \in JX \mid \chi^E_X(v) = \top_X(v)\}
    \]
    In other words, for any set \(X\), the following diagram is an equalizer diagram.
    \[\begin{tikzcd}
        EX \ar[r, hook, "e_X"] 
        & JX 
        \ar[r, "\chi^E_X", shift left]
        \ar[r, "\top_X"', shift right]
        & \Omega_\kappa X
    \end{tikzcd}\]
    It follows that \(e\colon E \subseteq J\) is an equational path constraint over \(\Omega_\kappa\).
\end{proof}

One takeaway of \cref{thm:converse to covariety} is that it is more informative to know that a path constraint is equational \emph{over a particular functor \(H\)} than to only know that it is equational.
This will come up again in \cref{thm:coax}. 

\begin{rem}
    The results above are suggestive of some deep connection between equational path constraints and the cardinalities of functors.
    But in fact, at the time of writing, the author is unaware of \emph{any} example of a non-equational path constraint in \(\Set\).
    For instance, as we will see below in \cref{eg:transitivity again}, transitivity is an equational path constraint even though \(\Pow\) is not accessible.
\end{rem}

\subsection{Modal formulas as equational path constraints}\label{subsec:frame conditions}
The equational rendering of transitivity for transition systems (henceforth, \emph{Kripke frames}) we will see below in \cref{eg:transitivity again} is an instance of a wider class of examples of equational path constraints over \(\Pow\) that are given by modal formulas (transitivity, for example, is captured by the modal formula \(\Diamond\Diamond p \to \Diamond p\)).
Since the connection to modal logic is interesting in its own right, we think it is worth taking some time to spell out explicitly.

We start with an interesting class of examples of modal formulas that define equational path constraints over the powerset functor, \(\Pow\).
Let \(\NatForm\) be the set of \emph{natural} modal formulas generated by the grammar \[
    \varphi,\psi ::= p \mid \bot \mid \varphi \vee \psi \mid \Diamond \varphi
\]
Here, \(p\) is a fixed basic proposition.
We define a functor \(K^\varphi\) and a natural transformation \(\sem{\varphi} \colon K^\varphi \Rightarrow \Pow\) for each \(\varphi \in \NatForm\) recursively by 
\begin{gather*}
    \begin{aligned}
        K^p &= K^\bot = \Id 
        & \sem{p}_X(x) &= \{x\}
        \quad \sem{\bot}_X(x) = \emptyset
        \\
        K^{\varphi \vee \psi} 
            &= K^\varphi \times K^\psi 
        & \sem{\varphi\vee \psi}_X(U, V) 
            &= \sem{\varphi}_X(U) \cup \sem{\psi}_X(V) 
        \\
        K^{\Diamond \varphi} 
            &= \Pow \circ K^\varphi
        & \sem{\Diamond \varphi}_X(U) 
            &= \bigcup \Pow (\sem{\varphi}_X)(U)
    \end{aligned}  
\end{gather*}
A routine induction argument shows that \(K^\varphi \in \Shape(\Pow)\) for all \(\varphi \in \NatForm\).

Let \(\mathbb P\) be the set of propositions, including the fixed proposition \(p\).
Recall that a Kripke \emph{model} is a structure \(\mathcal M = (X, \langle \nu, \beta\rangle)\) consisting of a \emph{valuation} \(\nu \colon X \to \Pow \mathbb P\) and a Kripke frame \(\beta \colon X \to \Pow X\).
For each \(x \in X\), the \emph{satisfaction relation} \(\mathcal M, x \models \varphi\) is defined recursively on \(\varphi\) as follows~\cite{BlackburnRV2001}:
\begin{itemize}
    \item[] (\emph{Base Case}) given \(q \in \mathbb P\), we have \(\mathcal M, x \models q\) if \(q \in \nu(x)\), 
    \item[] (\emph{Recursive Step}) given modal formulas \(\varphi, \varphi_1,\varphi_2\), we have
    \begin{itemize}
        \item[] \(\mathcal M, x \models \varphi_1 \vee \varphi_2\) if either \(\mathcal M \models \varphi_1\) or \(\mathcal M \models \varphi_2\)

        \item[] \(\mathcal M, x \models \neg \varphi\) if \(\mathcal M \not \models \varphi\)

        \item[] \(\mathcal M, x \models \Diamond \varphi\) if \(\exists y \in X\) such that \(y \in \beta(x)\) (i.e., \(x \to y\)) and \(\mathcal M, y \models \varphi\)
    \end{itemize}
\end{itemize}
The semantics of the operations \(\wedge\), \(\to\), and \(\Box\) are derived from the operations above via the identities \(\varphi \wedge \psi = \neg(\neg \varphi \vee \neg \psi)\), \(\varphi \to \psi = \neg \varphi \vee \psi\), and \(\Box \varphi = \neg \Diamond \neg \varphi\).
We write \(\mathcal M \models \varphi\) if \(\mathcal M, x \models \varphi\) for all \(x \in X\), and \(\models \varphi\) if \(\mathcal M \models \varphi\) for all Kripke models \(\mathcal M\).
Note that \(\mathcal M, x \not \models \bot\) always, and moreover \(\mathcal M \not\models \Diamond^n \bot\) for all \(n \ge 0\).

Given a Kripke frame \((X, \beta)\), we write \((X, \beta) \models \varphi\) if for any \(\nu \colon X \to \Pow \mathbb P\), the corresponding Kripke model satisfies  \((X, \langle \nu, \beta\rangle) \models \varphi\).
In this sense, modal formulas are properties that a \emph{frame} may or may not satisfy, so we use the phrase \emph{frame condition} as a synonym of \emph{modal formula}.
\smallskip

Again, to make it absolutely clear, \(\NatForm\) is built up from exactly one fixed basic proposition, which we have called \(p\).

\begin{prop}\label{prop:modal}
    For \(\varphi,\psi \in \NatForm\), the frame condition \(\varphi \leftrightarrow \psi\) is equivalent to the constraint \(\lambda \equiv \rho\) obtained from the transformations 
    \[
        \lambda = \sem{\varphi} \circ \proj_{K^\varphi} \colon K^\varphi \times K^\psi \Rightarrow \Pow
        \qquad 
        \rho = \sem{\psi} \circ \proj_{K^\psi} \colon K^\varphi \times K^\psi \Rightarrow \Pow
    \]
\end{prop}

\begin{proof}
    Let \((X, \beta)\) be a Kripke frame.
    A standard exercise in modal logic shows that \(\models \Diamond (\varphi_1 \vee \varphi_2) \leftrightarrow \Diamond \varphi_1 \vee \Diamond \varphi_2\).
    It is also true that every Kripke frame \((X, \beta)\) satisfies 
    \begin{equation}
        \sem{\Diamond (\varphi_1 \vee \varphi_2)}\circ \proj_{K^{\varphi_1}} 
        \equiv \sem{\Diamond \varphi_1 \vee \Diamond \varphi_2} \circ \proj_{K^{\varphi_2}}
    \end{equation}
    This amounts to a calculation:
    on the left, we have 
    \[
        K_1 := K^{\Diamond(\varphi_1 \vee \varphi_2)} = \Pow \circ (K^{\varphi_1} \times K^{\varphi_2})
    \] 
    and on the right 
    \[
        K_2 := K^{\Diamond\varphi_1 \vee \Diamond\varphi_2} = (\Pow \circ K^{\varphi_1}) \times (\Pow \circ K^{\varphi_2})
    \]
    Given a Kripke frame \((X, \beta)\), let \(x \in X\).
    Unravelling definitions, we obtain
    \begin{gather*}
        \beta^{K_1}(x) 
        = \Big\{ (\beta^{K^{\varphi_1}}(y), \beta^{K^{\varphi_2}}(y)) ~\Big|~ y \in \beta(x)\Big\}
        \\
        \beta^{K_2}(x) 
        = \Big(\Big\{ \beta^{K^{\varphi_1}}(y) ~\Big|~ y \in \beta(x)\Big\}, \Big\{ \beta^{K^{\varphi_2}}(z) ~\Big|~ z \in \beta(x)\Big\}\Big)
    \end{gather*}
    Thus,
    \begin{align*}
        &\sem{\Diamond (\varphi_1 \vee \varphi_2)} \circ \beta^{K_1}(x)\\
        &=  \bigcup \Big\{ \sem{\varphi_1} \circ \beta^{K^{\varphi_1}}(y)\cup \sem{\varphi_2} \circ \beta^{K^{\varphi_2}}(y) ~\Big|~ y \in \beta(x)\Big\} \\
        &=  \bigcup \Big\{ \sem{\varphi_1} \circ \beta^{K^{\varphi_1}}(y)\cup \sem{\varphi_2} \circ \beta^{K^{\varphi_2}}(z) ~\Big|~ y,z \in \beta(x)\Big\} 
        \tag{in general, \(\bigcup_{i \in I} (U_i \cup V_i) = \bigcup_{i,j \in I} (U_i \cup V_j)\)}\\
        &=  \bigcup \Big\{ \sem{\varphi_1} \circ \beta^{K^{\varphi_1}}(y) ~\Big|~ y \in \beta(x)\Big\} \cup \bigcup \Big\{\sem{\varphi_2} \circ \beta^{K^{\varphi_2}}(z) ~\Big|~ z \in \beta(x)\Big\} \\
        &=  \sem{\Diamond \varphi_1 \vee \Diamond \varphi_2} \circ \beta^{K_2}(x)
    \end{align*}
    It follows from above and a straightforward induction argument that, both logically and in terms of \(\sem{-}\), every formula \(\varphi \in \NatForm\) is equivalent to a formula of the form 
    \[
        \bigvee_{i=1}^k \Diamond^{n_i} p = \Diamond^{n_1}p \vee \Diamond^{n_2} p \vee \cdots \vee \Diamond^{n_k} p
    \]
    where \(p\) is the fixed basic proposition.
    Thus, given \(\varphi,\psi \in \NatForm\), without loss of generality, we can write 
    \begin{equation}
        \label{eq:explicit-natform}
        \varphi = \bigvee_{i=1}^k \Diamond^{n_i} p 
        \qquad
        \qquad
        \psi = \bigvee_{j=1}^k \Diamond^{m_i} p
    \end{equation}
    for some \(n_i, m_j, k \in \NN\) (this may involve repeating some terms).
    For the same \(\varphi,\psi\) above, by calculation we see that \(z \in \sem{\varphi}(x)\) if and only if there is an \(i \le k\) and a path \(x \to x_1 \to \cdots \to x_{n_i} = z\), and similarly for \(\psi\).\medskip

    \noindent
    We are now ready to address the core of the proposition.\smallskip

    \noindent%
    (\(\Rightarrow\)) Suppose \((X, \beta) \models \varphi \leftrightarrow \psi\), and let \(x \in X\) and \(z \in \sem{\varphi}\circ\beta^{K^\varphi}(x)\). 
    From the latter statement and \eqref{eq:explicit-natform}, we know there is an \(i \le k\) and a path \(x \to x_1 \to \cdots \to x_{n_i} = z\).
    Let \(\nu \colon X \to \Pow\mathbb P\) be the function defined \(\nu(z) = \{p\}\) and \(\nu(y) = \emptyset\) for all \(y \neq z\). 
    Then \((X, \langle \nu, \beta\rangle), x \models \varphi\), which by assumption implies \((X, \langle \nu, \beta\rangle), x \models \psi\).
    Thus, for some \(j \le k\), there is a path \(x \to y_1 \to \cdots \to y_{m_j} = z'\) such that \(p \in \nu(z')\). 
    Note, then, that \(z' \in \sem{\psi}\circ\beta^{K^\psi}(x)\).
    But only \(\nu(z) = \{p\}\) is nonempty, so \(z = z'\). 
    Consequently, \(z \in \sem{\psi}\circ\beta^{K^\psi}(x)\). 
    Since \(z\) and \(x\) were arbitrary, \(\sem{\varphi} \circ \beta^{K^\varphi} \subseteq \sem{\psi}\circ\beta^{K^\psi}\). 
    The claim is moreover symmetric in \(\varphi\) and \(\psi\), so the opposite inclusion is derived in the same way. 
    Hence, \(\sem{\varphi} \circ \beta^{K^\varphi} = \sem{\psi}\circ\beta^{K^\psi}\).
    This implies \((X, \beta)\) satisfies \(\sem{\varphi} \circ \proj_{K^\varphi} \equiv \sem{\psi} \circ \proj_{K^\psi}\).\smallskip

    \noindent%
    (\(\Leftarrow\)) Now suppose \((X, \beta)\) satisfies \(\sem{\varphi} \circ \proj_{K^\varphi} \equiv \sem{\psi} \circ \proj_{K^\psi}\).
    Then \(\sem{\varphi} \circ \beta^{K^\varphi} = \sem{\psi}\circ\beta^{K^\psi}\). 
    To establish the logical equivalence we are after, let \((X, \beta) \models \varphi\). 
    Let \(x \in X\) and \(\nu \colon X \to \Pow \mathbb P\), and note that by definition there is an \(i \le k\) such that \((X, \langle \nu, \beta\rangle), x \models \Diamond^{n_i} p\). 
    This implies there is a path \(x \to x_1 \to \cdots \to x_{n_i}\) with \(p \in \nu(x_{n_i})\).
    Let \(z = x_{n_i}\).
    The existence of the path mentioned tells us that \(z \in \sem{\varphi} \circ \beta^{K^\varphi}(x)\).
    Now, since \(\sem{\varphi} \circ \beta^{K^\varphi} = \sem{\psi}\circ\beta^{K^\psi}\), \(z \in \sem{\psi}\circ\beta^{K^\psi}(x)\). 
    By definition of \(\sem{\psi}\) and \eqref{eq:explicit-natform}, there is a \(j \le k\) and a path \(x \to y_1 \to \cdots \to y_{m_j} = z\). 
    Since \(\Diamond^{m_j} p\) is a summand of \(\psi\), we obtain \((X, \langle \nu, \beta\rangle), x \models \psi\). 
    But \(x\) and \(\nu\) were arbitrary, so \((X, \beta) \models \psi\), as desired. 
    This shows that \((X, \beta) \models \varphi \to \psi\). 
    From symmetry, we also obtain \((X, \beta) \models \varphi \leftrightarrow \psi\), as desired.
\end{proof}

The frame condition \(\varphi \to \psi\) is logically equivalent to \(\psi \vee \varphi \leftrightarrow \psi\), so actually, every frame condition of the form \(\varphi \to \psi\) with \(\varphi,\psi \in \NatForm\) is equational over \(\Pow\).

\begin{exa}
    \label{eg:transitivity again}
    Recall the path constraint \(E\) capturing transitivity in \cref{eg:transitivity}.
    The frame condition that captures transitivity is \(\Diamond\Diamond p \to \Diamond p\)~\cite{BlackburnRV2001}, the equational equivalent being \(\Diamond p \vee \Diamond\Diamond p \leftrightarrow \Diamond p\).
    By \cref{prop:modal}, this constraint is equational over \(\Pow\). 
    Unravelling the definitions, \(\lambda,\rho \colon \Pow \times \Pow^2 \Rightarrow \Pow\) are the natural transformations defined by \(\lambda_X(U, \Phi) = U\) and \(\rho_X(U, \Phi) = U \cup \bigcup \Phi\) for any set \(X\) and \((U, \Phi) \in \Pow X \times \Pow^2 X\). 
    In general, \(U_1 \subseteq U_2\) if and only if \(U_1 \cup U_2 = U_2\), so \((U, \Phi) \in EX\) (that is, \(\bigcup \Phi \subseteq U\)) if and only if 
    \[
        \lambda_X(U, \Phi)
        = U 
        = U \cup \bigcup \Phi 
        = \rho_X(U, \Phi) 
    \]
    It follows that \((X, \beta)\) is transitive if and only if \((X, \beta)\) satisfies \(\lambda \equiv \rho\).
    %
    %
\end{exa}

\begin{rem}
	\label{rem:frame}
    The reader might have noticed that symmetry, captured by the frame condition \(p \to \Box \Diamond p\), is not included in \cref{eg:transitivity again}.
    This is because \cref{eg:transitivity again} only shows that frame conditions of the form \(\varphi \leftrightarrow \psi\) for \(\varphi,\psi \in \NatForm\) can be expressed as equational path constraints.
    The issue with \(\wedge\), \(\Box\), and \(\neg\) is that \(\cap \colon (\Pow X)^2 \to \Pow X\) and complement \((-)^c \colon \Pow X \to \Pow X\) are not natural in \(X\).
    Despite this, symmetry can be expressed as an equational path constraint using a different method, and over \(\Pow^2\) instead of \(\Pow\).

    Given a Kripke frame \((X, \beta)\), symmetry can be expressed in terms of \(\beta\) as the property that for any \(x,y \in X\), \(y \in \beta(x)\) (i.e., \(x \to y\)) if and only if \(x \in \beta(y)\) (i.e., \(y \to x\)).
    Let us give an equivalent but slightly different phrasing: unravelling definitions, 
    \[
        \beta^2(x) = \{\beta(y) \mid y \in \beta(x)\} = \big\{\{z \mid y \to z\} ~|~ x \to y \big\}
    \]
    Then symmetry is equivalent to the identity
    \begin{equation}
        \label{eq:symm-phrasing1}
        \big\{\{z \mid y \to z\} ~|~ x \to y \big\} = \big\{\{x\} \cup \{z \mid y \to z\} ~|~ x \to y \big\}
    \end{equation}
    that ensures every \(y\) with a transition \(x \to y\) also has a transition \(y \to x\).
    In terms of \(\beta\), \eqref{eq:symm-phrasing1} becomes
    \[
        \{\beta(y) \mid y \in \beta(x)\} = \{\{x\} \cup \beta(y) \mid y \in \beta(x)\}
    \]
    This identity is something we can express as an equational path constraint.
    Define the operation \(\rho_X \colon X \times \Pow^2X \Rightarrow \Pow^2X\) by 
    \[
        \rho_X(x, \Phi) = \big\{\{x\} \cup U ~\big|~ U \in \Phi \big\}
    \]
    This is natural in \(X\), for given any \(f \colon X \to Y\) we have 
    \[
    \begin{tikzcd}
        X \times \Pow^2 X \ar[r, "\rho_X"] \ar[d, "f \times \Pow^2(f)"'] 
        & \Pow^2 X \ar[d, "\Pow^2(f)"] \\
        Y \times \Pow^2 Y \ar[r, "\rho_Y"] 
        & \Pow^2 Y \\
    \end{tikzcd}
    \qquad
    \begin{aligned}
        \Pow^2(f) \circ \rho_X(x, \Phi) 
        &= \big\{f(\{x\} \cup U) ~\big|~ U \in \Phi \big\} \\
        &= \big\{\{f(x)\} \cup f(U) ~\big|~ U \in \Phi \big\} \\
        &= \rho_Y (f(x), \{f(U) \mid U \in \Phi\}) \\
        &= \rho_Y (f(x), \Pow^2(f)(\Phi)) \\
        &= \rho_Y \circ (f \times \Pow^2(f))(x, \Phi)
    \end{aligned}
    \]
    Now, take \(J = \Id \times \Pow^2\), and let \(\lambda = \proj_{\Pow^2} \colon J \Rightarrow \Pow^2\).
    Then a frame \((X, \beta)\) satisfies \(\lambda \equiv \rho\) if and only if \eqref{eq:symm-phrasing1} holds, which we have already seen is equivalent to symmetry.
\end{rem}

Intuitively, it seems that symmetry is not, in any way, equational over \(\Pow\), but we have not yet found a formal proof of this.
We leave it open as to whether symmetry is equational over \(\Pow\) afterall, or whether there are any frame conditions that are not equational over \(\Pow\).

\subsection{Differential equations as equational path constraints}
Differential equations provide many interesting examples of equational path constraints.
In physics, a \emph{heat equation} in two spatial variables \(x,y\) and a time variable \(t\) is a partial differential equation of the form
\begin{equation}
    \label{eq:heat}
    \frac{df}{dt} = c\Big(\frac{d^2f}{dx^2} + \frac{d^2f}{dy^2}\Big)
\end{equation}
for some constant \(c \in \RR\).
One way to model this equation coalgebraically follows \cref{eg:commutative Moore automata} to the extent that it uses the same language of derivatives and input variables, but it requires scaling by a constant \(c\) and adding derivatives together. 
In other words, it requires a vector space structure. 

Let \(\Id \colon \VSp \to \VSp\) be the identity functor on the category of vector spaces and linear maps, and let \(A = \{t, x, y\}\).
The transformations \(\der a \colon \Id^A \Rightarrow \Id\) for \(a \in A\) are defined as in \cref{eg:commutative Moore automata} by \(\der a \theta = \theta(a)\).
Then \eqref{eq:heat} can be translated into an equational path constraint on \(\Id^A\)-coalgebras with 
\(\lambda, \rho \colon \Id^A \times (\Id^A)^A \Rightarrow \Id\) given by
\(
    \lambda = \der t \circ \proj_{\Id^A}
\) and \(
    \rho = c(\der x \der x + \der y \der y) \circ \proj_{(\Id^A)^A}
\).
These are indeed linear and natural.
Given a linear weighted automaton \((X, \langle o, \delta\rangle)\) (\cref{eg:examples of coalgebras}\eqref{eg:it:moore}), the path constraint \(\lambda \equiv \rho\) requires that
\[
    \der t \beta(x) = c(\der x \der x \beta^2(x) + \der y\der y \beta^2(x))
\]     
for any \(x \in X\).
For a concrete example, consider the vector space \(C^{\omega}\) of all analytic functions in three variables.
Then \(C^\omega\) obtains the structure of an \(\Id^A\)-coalgebra \((C^\omega, \beta)\) where \(\beta \colon C^\omega \to (C^\omega)^A\) is given by partial derivatives \(\beta(f)(a) = \frac{d}{da}f\). 
If we write \(C^\omega_{ht}\) for the set of analytic solutions to \eqref{eq:heat}, then \(C^\omega_{ht}\) is a subcoalgebra of \((C^\omega, \beta)\) that satisfies the equational path constraint \(\lambda \equiv \rho\). 

\subsection{Monoid presentations as equational path constraints}
A wealth of examples of equational path constraints can be obtained from \emph{monoid presentations}, systems of formal identities between words from an alphabet.
To begin with, let \(A^*\) be the set of words in the alphabet \(A\), and recall that the algebraic structure \((A^*, \cdot, 1)\) consisting of concatenation \({\cdot} \colon A^* \times A^* \to A^*\) and the empty word \(1 \in A^*\) is the \emph{free monoid generated by \(A\)}. 
That is, \((A^*, \cdot, 1)\) is a monoid, and for each function \(f \colon A \to M\) into a monoid \((M, *, 1_M)\) there is the unique monoid homomorphism \(f^\# \colon (A^*, \cdot, 1) \to (M, *, 1_M)\) such that \(f^\#(a) = f(a)\) for each \(a \in A \subseteq A^*\). 

Given a relation \(R \subseteq A^* \times A^*\) between words, let \(\approx_R\) be the smallest \emph{congruence} generated by \(R\), the smallest equivalence relation \(\approx\) such that \(R \subseteq {\approx}\) and if \(s_1 \approx t_1\) and \(s_2 \approx t_2\), then \(s_1 * t_1 \approx s_2 * t_2\). 
A monoid \((M, *, 1_M)\) is said to be \emph{presented by \(A\) and \(R\)}, and we write \((M, *, 1_M) = \langle A \mid R\rangle\), if there is an isomorphism \((M, *, 1_M) \cong (A^*/{\approx_R}, \star, [1]_{\approx_R})\).
Here, the binary operation of the monoid is given by \([w]_{\approx_R} \star [u]_{\approx_R} = [w u]_{\approx_R}\).

Now let \(F = \Delta_B \times \Id^A\) and \(R \subseteq A^* \times A^*\).
We translate \(R\) into a system of singular equational path constraints \(\Sys_R\) over \(\Id\) by extending the derivative transformations (\cref{eg:commutative Moore automata}) to all words, as follows:
given \(w = a_1 \cdots a_n \in A^*\), the transformation \(\der w \colon F^{\mathsf{len}(w)} \Rightarrow \Id\) is defined to be the composition
\(
    \der w = \der a_n \circ \cdots \circ \der a_1
\).

\begin{restatable}{prop}{monoidpresentation}
    \label{thm:monoid presentations}
    Let \(F = \Delta_B \times \Id^A\), \(R \subseteq A^* \times A^*\), and \((M, *, 1_M) = \langle A \mid R\rangle\).
    Define 
    \[
        \Sys_R = \{\der w \equiv \der u \mid (w,u) \in R\}
    \] 
    and let
    \(
        (Z_R, \langle o_R, \delta_R\rangle)
    \) 
    be the Moore automaton with \(Z_R = B^M\), 
    \(
        o_R(\theta) = \theta(1_M)
    \), 
    and 
    \(
        \delta_R(\theta)(a)(s) = \theta(s * a) 
    \) 
    for any \(\theta \in B^M\), \(a \in A\), and \(s \in M\).
    Then 
    \(
        (Z_R, \langle o_R, \delta_R\rangle)
    \)
    is final relative to  \(\Coalg(F, \Sys_R)\).
\end{restatable}

\begin{proof}
    We will assume, without loss of generality, that \((M, *, 1_M) = (A^*/{\approx_R}, \star, [\varepsilon]_{\approx_R})\).
    Let \((X, \langle o_X, \delta_X\rangle)\) be a Moore automaton.
    We show that the map \(\beh_{X} \colon X \to B^M\) defined by \(\beh_X(x)([w]_{\approx_R}) = o_X(\der w(x))\) is (1) well-defined, and (2) a unique homomorphism.
    
    For (1), we show that if \(w \approx_R u\), then \((\der w)_X(x) = (\der u)_X(x)\). 
    This is done by induction on the derivation of \(w \approx_R u\). 
    The only interesting steps in that induction proof are the base case, in which \((w, u) \in R\), and the congruence induction step.
    In the base case, \((\der w)_X(x) = (\der u)_X(x)\), because \((X, \langle o, \delta\rangle)\) satisfies \(R\) and \((w, u) \in R\). 
    In the induction step, we let \(w_1 \approx_R u_1\) and \(w_2 \approx_R u_2\), and assume for an induction hypothesis that \((\der w_i)_X(y) = (\der u_i)_X(y)\) for \(i=1,2\) and all \(y \in X\).
    We now have 
    \begin{align*}
        (\der w_1w_2)_X(x) 
        &= (\der w_2)_X \circ (\der w_1)_X(x) \\
        &= (\der u_2)_X \circ (\der w_1)_X(x) \tag{ind.~hyp.}\\
        &= (\der u_2)_X \circ (\der u_1)_X(x) \tag{ind.~hyp.}\\
        &= (\der {u_1u_2})_X (x) 
    \end{align*}
    So if \(w \approx_R u\), then
    \(\beh_X(x)([w]_{\approx_R}) = o_X((\der w)_X(x)) = o_X((\der u)_X(x)) = \beh_X(x)([u]_{\approx_R})\).

    For (2), we need to show that \(\beh_X\) is a homomorphism and that it is unique in that regard. 
    Starting with the homomorphism property, we need to show (i) \(o_R(\beh_X(x)) = o_X(x)\) and (ii) \(\delta_R(\beh_X(x))(a) = \beh_X(\delta_X(x)(a))\).
    For (i), 
    \[
        o_R(\beh_X(x)) = \beh_X(x)(\varepsilon) = o_X(\der \varepsilon~x) = o_X(x)
    \]
    For (ii), given \(w \in A^*\),
    \begin{align*}
        \delta_R(\beh_X(x))(a)([w]_{\approx_R})
        &= \beh_X(x)([a]_{\approx_R} \star [w]_{\approx_R}) \\
        &= \beh_X(x)([aw]_{\approx_R}) \\
        &= o_X((\der aw)_X(x)) \\
        &= o_X((\der w)_X\circ (\der a)_X (x)) \\
        &= o_X((\der w)_X\circ\delta_X(x)(a)) \\
        &= \beh_X(\delta_X(x)(a))([w]_{\approx_R})
    \end{align*}
    To see that \(\beh_X\) is the unique Moore automaton homomorphism, consider an arbitrary Moore automaton homomorphism \(h \colon (X, \langle o_X, \delta_X\rangle) \to (Z_R, \langle o_R, \delta_R\rangle)\).
    Then,
    \[
        h(x)([\varepsilon]_{\approx_R}) 
        = o_R(h(x))
        = o_X(x)
        = \beh_X(x)([\varepsilon]_{\approx_R})
    \]
    and inductively, for any \(a\in A\) and \(w \in A^*\), 
    \begin{align*}
        h(x)([aw]_{\approx_R})
        &= h(x)([a]_{\approx_R} \star [w]_{\approx_R}) \\
        &= \delta_R(h(x)([w]_{\approx_R}))(a) \\
        &= \delta_R(\beh_X(x)([w]_{\approx_R}))(a) \tag{ind.~hyp.} \\
        &= \beh_X(x)([a]_{\approx_R} \star [w]_{\approx_R})  \\
        &= \beh_X(x)([aw]_{\approx_R}) 
    \end{align*}
    Since \([-]_{\approx_R} \colon X \to M\) is surjective, \(\beh_X = h\).
\end{proof}

Under the same assumptions as in \cref{thm:monoid presentations}, we will say that a Moore automaton \((X, \beta)\) \emph{satisfies \(R\)} if \((X, \beta) \in \Coalg(F, \Sys_R)\).

\begin{exa}
	\label{eg:grids}
    Consider the monoid presentation \((\NN^2, +, (0,0)) = \langle a, b \mid R\rangle\) where \(a = (1,0)\), \(b = (0,1)\), and \(R = \{(ab, ba)\}\), i.e., encoding ``\(a + b = b + a\)''.
    Then \cref{thm:monoid presentations} tells us that the final Moore automaton satisfying \(R\) is carried by the set \(B^{\NN^2}\) of functions of the form \(\NN^2 \to B\). 
    The rest of the automaton structure on \(B^{\NN^2}\) is given by \(o(\theta) = \theta(0, 0)\) and \(\delta(\theta)(a)(i, j) = \theta(i + 1, j)\) and \(\delta(\theta)(b)(i, j) = \theta(i, j + 1)\) for any \(\theta \in B^{\NN^2}\) and \((i, j) \in \NN^2\). 
    Functions of the form \(\NN^2 \to B\) can be identified with formal power series in two commuting variables with coefficients in \(B\), so we obtain Boreale's characterization of the final commutative Moore automaton from~\cite{Boreale19} as an instance of \cref{thm:monoid presentations}.
\end{exa}

\begin{exa} 
    \label{eg:biinfinite streams}
    Now consider \((\ZZ, +, 0) = \langle a, b \mid R\rangle\) with \(a = 1\), \(b = -1\), and \(R = \{(ab, \varepsilon), (ba, \varepsilon)\}\). 
    One might call a Moore automaton that satisfies \(R\) \emph{reversible}, since every input action can be ``undone''.
    \cref{thm:monoid presentations} tells us that the final Moore automaton satisfying \(R\) is carried by \(B^\ZZ\). 
    Integer-indexed sequences from \(B\) are called \emph{bi-infinite streams}.
    Bi-infinite streams are a core topic in symbolic dynamics~\cite{Sternberg2010}, and are also studied as a certain kind of final coalgebra in~\cite{KupkeR08}.
\end{exa}

\subsection{Bounded branching as a path constraint}\label{subsec:bounded branching}
We would like to end this section with a somewhat tangential discussion, mostly to further illustrate the breadth of the path constraint idea.
Recall that \(\Met\) denotes the category of metric spaces and nonexpansive maps.
Given an endofunctor \(F \colon \Met \to \Met\), an \(F\)-coalgebra is a metric space \(X\) (we denote its metric with \(d^X \colon X \times X \to \RR_{\ge0}\)) paired with a nonexpansive map \(\beta \colon X \to FX\).
A natural type of constraint to consider in this context is that of \emph{bounded branching}.

\begin{defi}\label{def:bounded branching}
    Let \(F \colon \Met \to \Met\) be an endofunctor and \(\beta \colon X \to FX\) be an \(F\)-coalgebra.
    Given a shape \(J \in \Shape(F)\), an endofunctor \(H \colon \Met \to \Met\), natural transformations \(\lambda, \rho \colon J \Rightarrow H\), and \(\e > 0\), we say that \((X, \beta)\) has \emph{\(\e\)-bounded \((\lambda,\rho)\)-branching}, or alternatively that \((X, \beta)\) \emph{satisfies} \(\lambda \equiv_\e \rho\), if for any \(x \in X\),
    \[
        d^{FX}\big( \lambda_X \circ \beta^J(x), \rho_X \circ \beta^J(x) \big) < \e
    \]
    We write \(\Coalg(F, \lambda \equiv_\e \rho)\) for the full subcategory of \(\Coalg(F)\) consisting of \(F\)-coalgebras with \(\e\)-bounded \((\lambda,\rho)\)-branching.
\end{defi}

Natural examples of \(F\)-coalgebras with bounded branching come from topological dynamics, as mentioned in item~\eqref{eg:it:compact} of our list of example coalgebras.

\begin{exa}\label{eg:bounded rate}
    The endofunctor capturing topological dynamical systems as coalgebras is \(F = \Id \colon \Met \to \Met\).
    We can express a uniformly bounded ``rate of change'' in a topological dynamical system as follows: the shape we consider is \(J = \Id \times F \in \Shape(F)\) (where again, here \(F = \Id\)), together with \(H = \Id\), and \(\lambda = \proj_1 \colon \Id \times F \Rightarrow H\) and \(\rho = \proj_2 \colon \Id \times F \Rightarrow H\).
    Starting with a topological dynamical system \((X, \beta)\), i.e., with \(\beta \colon X \to F X = X\), its corresponding \(J\)-path coalgebra is \(\beta^J = \langle \id X, \beta\rangle \colon X \to X \times X\).
    Then the condition \(\lambda \equiv_\e \rho\) states that for any \(x \in X\), \(d^X(x, \beta(x)) < \e\). 
    Intuitively, \((X, \beta)\) evolves at a rate of less than \(\e\) in one time step.
    A plethora of more concrete examples comes to mind: translations \(\RR^n \to \RR^n\), rotations of the circle \(S^1 \to S^1\), etc.
\end{exa}

Perhaps unsurprisingly, bounded branching is a path constraint: given any endofunctors \(J,H \colon \Met \to \Met\), natural transformations \(\lambda,\rho \colon J \Rightarrow H\), a metric space \(X\), and an \(\e > 0\), let \(EX\) be the subspace of \(JX\) defined by
\[
    EX = \big\{ t \in JX \mid d^{HX}(\lambda_X(t), \rho_X(t)) < \e\big\}
\]
This is indeed a subfunctor of \(J\).
To see this, let \(f \colon X \to Y\) be a nonexpansive map and \(t \in EX\), and observe that
\begin{align*}
    d^{HY}(\lambda_Y \circ J(f)(t), \rho_Y \circ J(f)(t))
    &= d^{HY}(H(f) \circ \lambda_X(t), H(f) \circ \rho_X(t)) \\
    &\le d^{HX}(\lambda_X(t), \rho_X(t)) \\
    &< \e
\end{align*}
where we have made use of the nonexpansiveness of \(H(f)\) and that \(t \in EX \subseteq JX\).
What this calculation shows is that \(J(f)(t) \in EY\), so that \(J(f)\) restricts to the desired nonexpansive map \(E(f) \colon EX \to EY\).
Now, in particular, if \(J \in \Shape(F)\) for some endofunctor \(F\), then \(\lambda \equiv_\e\rho\) is captured by the path constraint \(E\) of shape \(J\).

\paragraph{Ordered branching} A similar kind of path constraint can be described for coalgebras on the category \(\Pos\) of partially ordered sets and monotone maps.

\begin{defi}\label{def:ordered branching}
    Let \(F \colon \Pos \to \Pos\) be an endofunctor and \(\beta \colon X \to FX\) be an \(F\)-coalgebra.
    Given a shape \(J \in \Shape(F)\), an endofunctor \(H \colon \Pos \to \Pos\), and natural transformations \(\lambda, \rho \colon J \Rightarrow H\), we say that \((X, \beta)\) has \emph{\((\lambda \le \rho)\)-branching}, or alternatively that \((X, \beta)\) \emph{satisfies} \(\lambda \le \rho\), if for any \(x \in X\),
    \[
        \lambda_X \circ \beta^J(x) \le^{HX} \rho_X \circ \beta^J(x)
    \]
    We write \(\Coalg(F, \lambda \le \rho)\) for the full subcategory of \(\Coalg(F)\) consisting of \(F\)-coalgebras with \((\lambda \le \rho)\)-branching.
\end{defi}

Again, for \(\lambda,\rho \colon J \Rightarrow H\), \((\lambda \le \rho)\)-branching is a path constraint for 
\[
    EX = \big\{ t \in JX \mid \lambda_X(t) \le^{HX} \rho_X(t)\big\}
\] 
The proof that \(E\) is a subfunctor of \(J\) is nearly identical to the proof that bounded branching is a path constraint.
Instead of spelling it out the proof, we move directly to a generalization that captures both the metric and ordered settings mentioned.

\paragraph{In an enriched setting}
A \emph{(unital) quantale} \((\mathcal V, \sqsubseteq, 1, \otimes)\) is a complete lattice with a monoid structure consisting of a \emph{unit} \(1\) and an associative binary operation \({\otimes} \colon \mathcal V \times \mathcal V \to \mathcal V\) that is component-wise continuous, in the sense that 
\[
    \Big(\bigsqcup U\Big) \otimes b = \bigsqcup_{a \in U} (a\otimes b)
    \qquad 
    b \otimes \Big(\bigsqcup U\Big) = \bigsqcup_{a \in U} (b\otimes a)
\]
for any \(U \subseteq \mathcal V\) and \(b \in \mathcal V\), where \(\bigsqcup\) denotes the join operation determined by the lattice order \(\sqsubseteq\).
See, e.g.,~\cite{Rosenthal1990} for an introduction to the theory of quantales.

A \emph{\(\mathcal V\)-category} is a set \(X\) equipped with an \emph{enriched hom function} \(\hom \colon X \times X \to \mathcal V\) that satisfies the following two conditions:
\begin{itemize}
    \item[] (\emph{identity}) 
    for any \(x \in X\), we have \(1 \le \hom(x, x)\)

    \item[] (\emph{composition}) 
    for any \(x, y, z \in X\), we have \(\hom(x, y) \otimes \hom(y, z) \sqsubseteq \hom(x, z)\)
\end{itemize}
A \emph{\(\mathcal V\)-functor} \(f \colon X \to Y\) between \(\mathcal V\)-categories \(X\) and \(Y\) is a function such that for any \(x,y \in X\), we have \(\hom^X(x, y) \sqsubseteq \hom^Y(f(x), f(y))\).
The category of all \(\mathcal V\)-categories and \(\mathcal V\)-functors is denoted \(\VCat\).

With the right choice in \(\VV\), extended pseudometric spaces (includes distance \(\infty\)) and preordered sets can be identified as \(\VV\)-categories~\cite{Lawvere73}.
Indeed, setting \(\VV = [0,\infty]\), \({\oplus} = {+}\), and \(r \sqsubseteq s\) iff \(r \ge s\), one obtains the category of extended pseudometric spaces and nonexpansive maps as \(\VCat\).
Note that in pseudometrics represented as quantales, upper bounds in the metric are represented by lower bounds in the quantale.
Similarly, setting \(\VV = 2\), \({\otimes} = \min\), with \(1\) as the unit, one obtains the category of preorders and monotone maps as \(\VCat\); here, the order \(x \le y\) is expressed as \(1 \sqsubseteq \hom(x, y)\).

There has been increasing interest in coalgebras in \(\VCat\)~\cite{WORRELL2000337,Resende01,BilkovaKPV13,BalanKV15,Kori0RK24,dangelo_et_al2024,BeoharG0MFSW24}, as these coalgebras enjoy quantale-valued behavioural metrics that capture behavioural distance and similarity in a unified setting.
In the same way, bounded branching and ordered branching can be captured as a kind of bounded branching where the bound comes from the underlying quantale.

\begin{defi}\label{def:enriched branching}
    Let \(F \colon \VCat \to \VCat\) be an endofunctor and \(\beta \colon X \to FX\) be an \(F\)-coalgebra.
    Given a shape \(J \in \Shape(F)\), an endofunctor \(H \colon \VCat \to \VCat\), a \emph{bound} \(b \in \VV\), and natural transformations \(\lambda, \rho \colon J \Rightarrow H\), we say that \((X, \beta)\) has \emph{\(b\)-bounded \((\lambda, \rho)\)-branching}, or alternatively that \((X, \beta)\) \emph{satisfies} \(\lambda \equiv_b \rho\), if for any \(x \in X\),
    \[
        b \sqsubseteq \hom(\lambda_X \circ \beta^J(x), \rho_X \circ \beta^J(x))
    \]
    We write \(\Coalg(F, \lambda \equiv_b \rho)\) for the full subcategory of \(\Coalg(F)\) consisting of \(F\)-coalgebras with \(b\)-bounded \((\lambda, \rho)\)-branching.
\end{defi}

As in the previous two cases, \(b\)-bounded \((\lambda,\rho)\)-branching can be shown to be a path constraint of shape \(J\) given by
\[
    EX = \big\{t \in JX \mid b \sqsubseteq \hom(\lambda_X(t), \rho_X(t))\big\}
\]
Indeed, given a \(\VV\)-functor \(f \colon X \to Y\) and a \(t \in EX\),
\begin{align*}
    b 
    &\sqsubseteq \hom^{HX}(\lambda_X \circ \beta(t), \rho_X \circ \beta(t)) \\
    &\sqsubseteq \hom^{HY} (H(f) \circ \lambda_X \circ \beta(t), H(f) \circ \rho_X \circ \beta(t)) \\
    &\sqsubseteq \hom^{HY} (\lambda_Y \circ J(f) \circ \beta(t), \rho_Y \circ J(f) \circ \beta(t))
\end{align*}
so that \(J(f)(t) \in EY\) as well.
This generalizes the argument in the metric case above, and shows that \(E\) is a subfunctor of \(J\).
In other words, \(\lambda \equiv_b \rho\) defines a path constraint.

In anticipation of \cref{thm:coax}, we can show the following closure properties for bounded branching in an enriched setting.

\begin{thm}\label{thm:bounded enriched branching}
    Let \(\VV\) be a quantale, \(F,H \colon \VCat \to \VCat\) endofunctors, \(J \in \Shape(F)\), and \(\lambda,\rho \colon J \Rightarrow H\).
    Let \(b \in \VV\).
    Then \(\Coalg(F, \lambda \equiv_b \rho)\) is a covariety (see \cref{def:covariety}).
\end{thm}

\begin{proof}
    We will see closure under coproducts and homomorphic images in~\cref{thm:coax}, so it suffices to show that \(\Coalg(F, \lambda \equiv_b \rho)\) is closed under the subcoalgebra relation (which does not appear to follow from~\cref{thm:coax} in this case).
    This is relatively straightforward: given \((X, \beta) \in \Coalg(F, \lambda \equiv_b \rho)\) and a subcoalgebra \(\incl_U \colon (U, \beta_U) \hookrightarrow (X, \beta)\), let \(x \in U\).
    Then since \(\incl_U(x) \in X\),
    \begin{align*}
        b 
        &\sqsubseteq \hom^{HX}(\lambda_X \circ \beta^J \circ \incl_U(x), \rho_X \circ \beta^J \circ \incl_U(x)) \\
        &\sqsubseteq \hom^{HX}(\lambda_X \circ J(\incl_U) \circ \beta_U^J(x), \rho_X \circ J(\incl_U) \circ \beta_U^J(x)) \tag{\cref{lem:always a homom}}\\
        &\sqsubseteq \hom^{HX}(H(\incl_U) \circ \lambda_Y \circ \beta_U^J(x), H(\incl_U) \circ \rho_Y \circ \beta_U^J(x)) \\
        &\sqsubseteq \hom^{HX}(\lambda_Y \circ \beta_U^J(x), \rho_Y \circ \beta_U^J(x)) 
    \end{align*}
    Since \(x \in U\) was arbitrary, \((U, \beta_U)\) also has \(b\)-bounded \((\lambda,\rho)\)-branching.
\end{proof}

\section{Equational covarieties, coequations, and chromatic bounds}
\label{sec:covariety}
Now that we have seen some examples of path constraints, we are ready to develop their basic theory.
There are two main goals of this section: the first is to show that the internal structure of \(\Coalg(F, \Sys)\) for a system of equational path constraints \(\Sys\) allows for standard constructions in coalgebraic modelling (i.e., that of a covariety).
In the category of sets and functions, this implies the existence of a coaxiomatization of the path constraint (and in fact, we will establish that the converse also holds for finitary functors on \(\Set\)). 
The second goal is to give a concrete coaxiomatization obtained in the \(\Set\) case for equational path constraints with some additional structure (this is the content of \cref{thm:coax}). 
\medskip

Let us begin with the first goal.

\begin{restatable}{thm}{thmpathcovariety}
    \label{thm:path covariety}
    Let \(\Sys\) be a system of path constraints.
    Then \(\Coalg(F, \Sys)\) is closed under coproducts and homomorphic images.
    If \(\Sys\) is equational, then \(\Coalg(F, \Sys)\) is also closed under the subcoalgebra relation, i.e., \(\Coalg(F, \Sys)\) is a covariety (see \cref{def:covariety}).
\end{restatable}

\begin{proof}
    Towards closure under coproducts, let \((X_i, \beta_i)\) be an \(F\)-coalgebra satisfying \(\Sys\) for each \(i \in I\).
    Let \(e \colon E \subseteq J\) be a path constraint in \(\Sys\), and let \(e_{X_i} \circ \e_i \colon X_i \to EX_i \to JX_i\) be the factorization of \(\beta_i^J\) for each \(i \in I\).
    Then by composing with \(E(\incl_i)\), we obtain an arrow \(E(\incl_i) \circ \e_i \colon X_i \to EX_i \to EY\) for each \(i \in I\).
    This induces an arrow \([E(\incl_i) \circ \e_i]_{i \in I} \colon Y \to EY\).
    Since for any \(i \in I\) we have
    \begin{equation*}
        \begin{tikzcd}
            X_i \ar[r, "\incl_i"] \ar[rr, bend left, "J(\incl_i) \circ \beta_i^J"] \ar[d, "\e_i"]
            & Y \ar[r, "\gamma^J"] \ar[d, dashed]
            & J Y
            \\
            EX_i \ar[rr, bend right, "e_{X_i}"] \ar[r, "E(\incl_i)"]
            & EY \ar[ur, "e_Y"]
            & JX_i \ar[u, "J(\incl_i)"']
        \end{tikzcd}
        \quad
        \begin{aligned}
            &\gamma^J \circ \incl_i \\
            &= J(\incl_i) \circ \beta_i^J & \text{(\cref{lem:always a homom})}\\
            &= J(\incl_i) \circ e_{X_i} \circ \e_i &\text{(factoring \(\beta_i^J\))}\\
            &= e_{Y} \circ E(\incl_i) \circ \e_i &\text{(\(e\) natural)}
        \end{aligned}
    \end{equation*}
    we conclude that \(\gamma^J = e_Y \circ [E(\incl_i) \circ \e_i]_{i \in I}\). 
    Therefore, \(\gamma^J\) factors through \(EY\), so \((Y, \gamma)\) satisfies \(E\). 
    Since \(e \colon E \subseteq J\) in \(\Sys\) was arbitrary, \((Y,\gamma) \in \Coalg(F, \Sys)\).

    For closure under homomorphic images, let \((X, \beta) \in \Coalg(F, \Sys)\), \(q \colon (X, \beta) \to (Y, \gamma)\) such that \(q \colon X \twoheadrightarrow Y\) is a regular epi, and let \(e \colon E \subseteq J\) be a path constraint in \(\Sys\). 
    Consider the diagram below. 
    \begin{equation}
        \begin{tikzcd}
            Q 
                \ar[r, shift left, "p"]
                \ar[r, shift right, "r"']
            & X
                \ar[r, "\e_\beta"] 
                \ar[rr, bend left, "\beta^J"] 
                \ar[d, two heads, "q"] 
            & EX 
                \ar[r, hook, "e_X"]
                \ar[d, "E(q)"] 
            & {JX} 
                \ar[d, "J(q)"] 
            \\
            & Y 
                \ar[r, dashed, "\e_\gamma"]
                \ar[rr, bend right, "\gamma^J"]
            & EY 
                \ar[r, hook, "e_Y"]
            & {JY} 
        \end{tikzcd}
    \end{equation}
    Above, \(q\) is a coequalizer of \(p,r \colon Q \rightrightarrows X\).
    Our goal is to induce the dashed arrow using the universal property of \(q\).
    This can be established using the monicity of \(e_Y\).
    From the calculation
    \(
        e_Y \circ E(q) \circ \e_\beta 
        = J(q) \circ e_X \circ \e_\beta 
        = J(q) \circ \beta^J 
        = \gamma^J \circ q 
    \)
    it follows that
    \begin{equation}
        e_Y \circ E(q) \circ \e_\beta \circ p
        = \gamma^J \circ q \circ p
        = \gamma^J \circ q \circ r
        = e_Y \circ E(q) \circ \e_\beta \circ r
    \end{equation}
    It follows from \(e_Y\) being monic that \(E(q) \circ \e_\beta \circ p = E(q) \circ \e_\beta \circ r\).
    Now, since \(q\) is a coequalizer of \(p\) and \(r\), we obtain a unique arrow \(\e_\gamma \colon Y \to EY\) satisfying the identity 
    \(
       \e_\gamma \circ q = E(q) \circ \e_\beta
    \).
    Thus, 
    \[
        e_Y \circ \e_\gamma \circ q
        = J(q) \circ \beta^J
        = \gamma^J \circ q
    \]
    Since \(q\) is epi, \(e_Y \circ \e_\gamma = \gamma^J\), as desired.
    Therefore \((Y, \gamma) \in \Coalg(F, \Sys)\) as well.

    For closure under subcoalgebras, let \(\incl \colon (U, \beta_U) \hookrightarrow (X, \beta)\) be a subcoalgebra and \(e \colon E \subseteq J\) be an equational path constraint in \(\Sys\) with left and right transformations \(\lambda,\rho \colon J \Rightarrow H\). 
    We obtain the following diagram.
    \begin{equation}
        \begin{tikzcd}
            U 
                \ar[r, "{\beta_U^J}"] 
                \ar[d, hook, "\incl"'] 
            & {J U} 
                \ar[r, shift left, "\lambda_U"] 
                \ar[r, shift right, "\rho_U"'] 
                \ar[d, "{J(\incl)}"']  
            & HU 
                \ar[d, "{H(\incl)}"]
            \\
            X 
                \ar[r, "\beta^J"] 
            & {J X}
                \ar[r, shift left, "\lambda_X"] 
                \ar[r, shift right, "\rho_X"']  
            & HX
        \end{tikzcd}
    \end{equation}
    Then we have
    \begin{align*}
        H(\incl) \circ \lambda_U \circ \beta_U^J
        &= \lambda_X \circ J(\incl) \circ \beta_U 
        \tag{naturality of \(\lambda\)}\\
        &= \lambda_X \circ \beta^J \circ \incl 
        \tag{\cref{lem:always a homom}}\\
        &= \lambda_X \circ e_X \circ \e_\beta \circ \incl 
        \tag{factorization of \(\beta^J\)}\\
        &= \rho_X \circ e_X \circ \e_\beta \circ \incl 
        \tag{equational over \(H\)}\\
        &= H(\incl) \circ \rho_U \circ \beta_U^J
        \tag{previous steps in reverse}
    \end{align*}
    We have assumed that \(H\) preserves monics, so \(H(\incl)\) is monic. 
    It follows that \(\rho_U \circ \beta_U^J = \lambda_U \circ \beta_U^J\). 
    By \cref{lem:equational}, \((U, \beta_U)\) satisfies \(E\).
\end{proof}

If \(\Cat = \Set\) and \(F\) is accessible, then \cref{thm:path covariety,thm:acc set eq} imply that every system of path constraints defines a covariety.
With the covariety theorem of Rutten (\cref{thm:rutten-gumm}) for endofunctors on \(\Set\), we obtain the following. 

\begin{cor}
    \label{corr:coaxiomatization}
    Let \(F \colon \Set \to \Set\) be accessible.
    Then for any system of path constraints \(\Sys\), there is a set \(K\) of colours with a cofree coalgebra \((Z_K, \langle\ell, \zeta_K\rangle)\) and a coequation \(C \subseteq Z_K\) such that \(\Coalg(F, \Sys) = \Coalg(F; C)\) (see \cref{def:coequation}).
\end{cor}

\subsection{Fintary covarieties are equational}
Finitary functors are a very special case in this story. 
It turns out that for finitary functors, the converse of \cref{thm:path covariety} also holds. 
The proof makes use of the terminal sequence (see \cref{par:terminal}).

\begin{thm}
    \label{thm:converse to covariety}
    Let \(F \colon \Set \to \Set\) be a finitary endofunctor, and let \(\Cov \subseteq \Coalg(F)\) be a covariety. 
    Then there is a shape \(J \in \Shape(F)\) and an equational path constraint \(E \subseteq J\) such that \(\Cov = \Coalg(F, E)\).
\end{thm}

\begin{proof}
    By \cref{thm:rutten-gumm}, there is a colour palette \(K\) and a coequation \(C \subseteq K \times Z_K\) such that \(\Cov = \Coalg(F; C)\) (here, \((Z_K, \langle \ell, \zeta_K\rangle)\) is a cofree \(F\)-coalgebra in \(K\) colours).
    Our goal will be to turn \(C\) into a path constraint of a particular shape \(J \in \Shape(F)\) that captures the covariety \(\Cov\).
    This can be done using the terminal sequence of \(\Delta_K \times F\), which constructs a cofree \(F\)-coalgebra in \(K\) colours.
    It suffices to simply construct a path constraint (rather than the left and right transformations of an equational path constraint) because every path constraint for an accessible functor (and in particular every finitary functor) is equational by~\cref{thm:acc set eq}.

    Let \(S \colon \Ord^\op \to \Set\) be the terminal sequence for the functor \(\Delta_K \times F \colon \Set \to \Set\) (replace \(F\) in \cref{def:terminal sequence} with \(\Delta_K \times F\)).
    Since \(F\) is finitary, the functor \(\Delta_K \times F\) is also finitary.
    By the second part of the theorem of Worrell (\cref{thm:worrell}) the coalgebra \((S_{\omega2}, (\bang^{\omega2+1}_{\omega2})^{-1})\) is a final \(\Delta_K\times F\)-coalgebra---in other words, it is a cofree \(F\)-coalgebra in \(K\) colours. 
    Thus, we can take \(Z_K = S_{\omega2}\) and \(\langle \ell, \zeta_K\rangle = (\bang^{\omega2+1}_{\omega2})^{-1}\) and observe that our coequation \(C\) is then a subset of \(S_{\omega2}\).
    Furthermore, by the first part of \cref{thm:worrell}, \(\bang^{\omega2}_{\omega}\) is injective. 
    This allows us to further identify \(C\) with a subset \(C \subseteq S_{\omega}\) (given the same name).

    Now, \((S_\omega, \bang^\omega)\) is a limit of a cochain in \(\Set\).
    These have a canonical description~\cite{Riehl16}: we can take \(S_\omega\) to be the set of all \emph{compatible sequences},
    \[
        S_{\omega} = \Big\{ (\sigma_n)_{n < \omega} \in \prod_{n < \omega} S_n \Bigmid \forall i < j,~\sigma_i = \bang^{j}_i(\sigma_j)\Big\}
    \]
    This motivates our choice in shape:
    recursively construct the shapes \(J_n \in \Shape(F)\), by setting \(J_0 = \Id\) in the base case, and  \(J_{n+1} = \Id \times (F\circ J_n)\) for any \(n \in \NN\) in the recursive case.
    Our chosen shape is then the product of these, \(J = \prod_{n < \omega} J_n\).
    The desired path constraint will be a subfunctor of \(J\).

    Now, given any set \(X\), and any function \(k \colon X \to K\), we obtain the sequence of functions \(\hat k_n \colon J_n X \to S_n\) defined recursively by setting \(\hat k_0 = \bang \colon X \to S_0 (= 1)\) and for any \(n < \omega\), by setting \(\hat k_{n+1} \colon J_{n+1}X \to S_{n+1}\) to be the map \(\hat k_{n+1} = k \times F(\hat k_n)\) (spelling out the types, \(J_{n+1}X = X \times FJ_nX\) and \(S_{n+1} = K \times FS_n\)).
    Then, for any set \(X\), we define \(EX\) to be all the sequences in \(JX\) whose images under recolourings land in the components of \(C\). 
    That is,
    \begin{equation}
        \label{eq:subfunctor in suff}
        EX = \Big\{(\rho_n)_{n < \omega} \in JX \Bigmid \forall k \colon X \to K,~\hat k_n(\rho_n) \in \bang^{\omega}_n(C)\Big\}
    \end{equation}
    Our remaining obligations are to show (1) that \(E\) is a subfunctor of \(J\), and (2) that for any \(F\)-coalgebra \((X, \beta)\), \((X, \beta)\) satisfies \(C\) if and only if \((X, \beta)\) satisfies \(E\).
    For (1), we will need the following identities.

    \medskip\noindent\textbf{Claim 1.} For any \(l \colon Y \to K\), \(f \colon X \to Y\), and \(n < \omega\), \(\hat l_n \circ J_n(f) = \widehat{l \circ f}_n\).

    \smallskip\noindent\emph{Proof of Claim 1.} By induction on \(n\). 
    In the base case, both \(\widehat{l \circ f}_0\) and \(\widehat l_0 \circ J_0(f)\) are arrows with codomain \(1\), so they must be equal.
    In the inductive step, assume that \(\widehat{l \circ f}_n = \widehat l_n \circ J_n(f)\). 
    Then we have
    \begin{align*}
        \widehat{l \circ f}_{n + 1}
        &= (l \circ f) \times F(\widehat{l \circ f}_n) \tag{def.} \\
        &= (l \circ f) \times F(\widehat l_n \circ J_n(f)) \tag{ind.~hyp.} \\
        &= (l \circ f) \times (F(\widehat l_n) \circ FJ_n(f)) \tag{functoriality} \\
        &= (l \times F(\hat l_n)) \circ (f \times FJ_n(f)) \tag{def.~of \(\times\)}\\
        &= \hat l_{n+1} \circ FJ_{n+1}(f)  \tag{def.}
    \end{align*}
    This concludes the proof of Claim 1.
    \medskip

    To see (1), now, let \(f \colon X \to Y\) and let \((\rho_n)_{n < \omega} \in EX\). 
    We need to show that \(J(f)((\rho_n)_{n < \omega}) \in EY\).
    To that end, let \(l \colon Y \to K\) be any colouring of \(Y\).
    Then for any \(n \in \NN\), by Claim 1,
    \begin{align*}
        \hat{l}_n \circ J_n(f)((\rho_n)_{n < \omega}) 
        = \widehat{l \circ f}_n((\rho_n)_{n < \omega})
        \in \bang^\omega_n(C)
    \end{align*}
    since \((\rho_n)_{n < \omega} \in EX\).
    Since \(l\) was arbitrary, we see that \(J(f)((\rho_n)_{n < \omega}) \in EY\) as desired. 

    Let us move on to (2).
    Let \((X, \beta)\) be any \(F\)-coalgebra.
    We need to show that \((X, \beta)\) satisfies \(C\) if and only if it satisfies the path constraint \(E\) defined above.
    To that end, given any colouring \(k \colon X \to K\), we obtain the unique \(\Delta_K \times F\)-coalgebra homomorphism \(h \colon (X, \langle k, \beta\rangle) \to (Z_K, \langle \ell, \zeta_K\rangle)\) from the terminal sequence (in essentially the same way that appears in the proof of \cref{thm:worrell} in~\cite{Worrell98} and of the related theorem of Barr~\cite{Barr93}).
    
    Given \(k \colon X \to K\), we obtain a cone \((X, \overline k)\) for \(S\) as follows: we define the map 
    \(\overline k_0 = \bang \colon X \to 1\) 
    and the map \(\overline k_{n+1} \colon X \to S_{n+1}\) by 
    \(\overline k_{n+1} = k \times (F(\overline k_n) \circ \beta)\) 
    for any \(n < \omega\).
    Since we necessarily have \(\bang^{1}_0 \circ \overline k_1 = \bang\), and for any \(n <\omega\) we also obtain 
    \begin{align*}
        \bang^{n+2}_{n+1} \circ \overline k_{n+2}
        &= (\id K \times F(\bang^{n+1}_n)) \circ (k \times (F(\overline k_{n+1}) \circ \beta)) \tag{def.}\\
        &= k \times (F(\bang^{n+1}_n) \circ F(\overline k_{n+1}) \circ \beta) \tag{def.~of \(\times\)}\\
        &= k \times (F(\bang^{n+1}_n \circ \overline k_{n+1}) \circ \beta) \tag{functoriality}\\
        &= k \times (F(\overline k_n) \circ \beta) \tag{ind.~hyp.} \\
        &= \overline k_{n+1} \tag{def.}
    \end{align*}
    by induction, \((X, \overline k)\) is so far a cone for the \(\omega\)-segment of \(S\).
    This allows us to take \(\overline k_{\omega} \colon (X, \overline k) \to (S_\omega, \bang^\omega)\) to be the unique cone homomorphism.
    By repeating this construction, one obtains a cone \((X, \overline k)\) for the entirety of \(S\), although we will only need its \((\omega2 + 1)\)-segment.
    In particular, the cone homomorphism \(\overline k_{\omega2} \colon (X, \overline k) \to (S_{\omega2}, \bang^{\omega2})\) (here, \(k\) is restricted to \(\omega2\)) is the unique \(\Delta_K\times F\)-coalgebra homomorphism \(\overline k_{\omega2} = \langle k, \beh_k\rangle \colon (X, \langle k, \beta\rangle) \to (Z_K, \langle \ell, \zeta_K\rangle)\).
    It follows that \((X, \beta)\) satisfies \(C\) if and only if for any \(k \colon X \to K\) and \(x \in X\), \(\overline k_{\omega2}(x) \in C\).

    Again, \(\bang^{\omega2}_\omega \colon S_{\omega2} \to S_\omega\) is injective and we identified \(C\) with a subset of \(S_\omega\) as a result.
    From this we see that \((X, \beta)\) satisfies \(C\) if and only if for any \(k \colon X \to K\) and \(x \in X\), \(\overline k_\omega(x) \in C\). 
    Now, \(C \subseteq S_\omega \subseteq \prod_{n \in \NN} S_n\).
    We can therefore write each pattern \(\sigma \in C\) as a compatible sequence, \(\sigma = (\sigma_n)_{n \in \NN}\).
    It follows from this and the definition of \(\bang^\omega\) that for each \(n \in \NN\), \(\sigma_n = \bang^\omega_n(\sigma)\).
    From this we see that \((X, \beta)\) satisfies \(C\) if and only if for any \(k \colon X \to K\) and any \(x \in X\) and any \(n \in \NN\), \(\bang^\omega_n \circ \overline k_\omega(x) = \overline k_n(x) \in \bang_n^\omega(C)\).
    But we also have the following identities: for any \(n \in \NN\) and \(k \colon X \to K\), \(\hat k_n \circ \beta^{J_n} = \overline k_n\).
    This can be proven inductively, as follows: we have
    \(
        \hat k_0 \circ \beta^{J_0}
        = \bang = \overline k_0
    \)
    in the base case, and in the inductive case,
    \begin{align*}
        \hat k_{n+1} \circ \beta^{J_{n+1}}
        &= (k \times F(\hat k_n)) \circ \beta^{\Id \times (F \circ J_{n})} \tag{def.}\\
        &= (k \times F(\hat k_n)) \circ \langle \id X, F(\beta^{J_n}) \circ \beta\rangle \tag{\cref{def:path constraint}} \\
        &= k \times (F(\hat k_n) \circ F(\beta^{J_n}) \circ \beta) \tag{def.~of \(\times\)} \\
        &= k \times (F(\hat k_n \circ \beta^{J_n}) \circ \beta) \tag{functoriality} \\
        &= k \times (F(\overline k_n) \circ \beta) \tag{ind.~hypothesis} \\
        &= \overline k_{n+1} \tag{def.}
    \end{align*}
    By the definition of \(E\), \((X, \beta)\) satisfies \(E\) if and only if for any \(k \colon X \to K\) and any \(n \in \NN\), \(\hat k_n \circ \beta^{J_n}\) factors through \(\bang^\omega_n(C)\).
    Since \((X, \beta)\) satisfies \(C\) if and only if every \(\overline k_n\) factors through \(\bang^\omega_n(C)\) for all \(n \in \NN\) and \(\overline k_n = \hat k_n \circ \beta^{J_n}\), we can finally conclude that \((X, \beta)\) satisfies \(E\) if and only if \((X, \beta)\) satisfies \(C\). 
\end{proof}

To summarize the two results, \cref{thm:path covariety} and \cref{thm:path covariety} tell us that every covariety of \(F\)-coalgebras for a finitary endofunctor \(F \colon \Set \to \Set\) is an equational covariety, i.e., axiomatized by an equational path constraint.
One might be tempted to say that equational path constraints are therefore the whole story when it comes to covarieties for a finitary endofunctor. 
But again, the shapes and functors over which the path constraints constructed in the proofs of these theorems are, in many examples, overkill.
It is much more informative to determine that a covariety is axiomatized by a particular equational path constraint, of a particular shape, over a particular functor.
The following subsection illustrates this.

\subsection{A bound on the chromatic numbers of equational covarieties}
The covariety theorem of Rutten (\cref{thm:rutten-gumm}) guarantees the existence of a coequation that coaxiomatizes a given covariety of \(F\)-coalgebras for an accessible endofunctor \(F \colon \Set \to \Set\). 
Its proof shows that there is a set of all possible behaviour patterns in a large enough set of colours, and takes that set to be the coequation. 
This description of the coequation is abstract by design, and as a result cannot be expected to produce familiar accounts of behavioural properties like transitivity or commutativity.
As a result, \cref{corr:coaxiomatization} does not say much about the coequation \(C\) coaxiomatizing \(\Sys\), only that it exists.

The first main result of the paper, \cref{thm:coax} below, provides more information than \cref{corr:coaxiomatization}: it gives an upper bound on the chromatic number of the covariety determined by an equational path constraint over an accessible endofunctor \(H\).
An additional assumption about \(H\) is required: to get a meaningful bound, we need \(H\) to carry an algebraic structure.

\begin{defi}
    \label{def:monad}
    A triple \((H, \eta, \mu)\) consisting of an endofunctor \(H\) on \(\Cat\) and two natural transformations \(\eta \colon \Id \Rightarrow H\) and \(\mu \colon HH \Rightarrow H\) is called a \emph{monad} if for any object \(X\) we have both \(\mu_X \circ \eta_{HX} = \id{HX} = \mu_X \circ H(\eta_X)\) and \(\mu_X \circ \mu_{HX} = \mu_X \circ H(\mu_X)\).
    An \emph{Eilenberg-Moore (EM-)algebra} for \((H, \eta, \mu)\) is a pair \((A, \alpha)\) consisting of an object \(A\) and an arrow \(\alpha \colon HA \to A\) such that \(\alpha \circ \eta_A = \id{A}\) and \(\alpha \circ \mu_A = \alpha \circ H(\alpha)\). 
    An \emph{algebra homomorphism} \(h\colon (A, \alpha) \to (B, \vartheta)\) is an arrow \(h \colon A \to B\) such that \(\vartheta \circ H(h)= h \circ \alpha\).
\end{defi}

Every pair of the form \((HX, \mu_X)\) is an EM-algebra for the monad \((H, \eta,\mu)\).
In fact, \((HX, \mu_X)\) is the \emph{free} EM-algebra on \(X\):
given an EM-algebra \((A, \alpha)\) and an arrow \(f \colon X \to A\), there is a unique algebra homomorphism \(f^\#\colon (HX, \mu_X) \to (A,\alpha)\) such that \(f^\# \circ \eta_X = f\). 
Specifically, \(f^\# = \mu_A \circ H(f)\).
This universal property is what we need in the construction of our coaxiomatization. 

\begin{restatable}[Coaxiomatization]{thm}{coaxiomatizationtheorem}%
    \label{thm:coax}
    Let \(F\colon \Set \to \Set\), \(\kappa\) be a cardinal, and \((H, \eta, \mu)\) be a monad with \(\locCard{H} \le \kappa\) (\cref{def:small}).
    Let \(K = H(\kappa + \kappa)\) and assume that there is a final \(\Delta_K\times F\)-coalgebra \((Z_K, \langle \ell_K, \zeta_K\rangle)\).
    Given a system of equational path constraints \(\Sys\) over \(H\), define the coequation 
    \begin{equation}
        \label{def:coax}
        C = \big\{ b \in Z_K ~\big|~ \forall \text{\(J\)-path const. } (\lambda \equiv \rho) \in \Sys,~ \ell_K^\# \circ \lambda \circ \zeta_K^J(b) = \ell_K^\# \circ \rho \circ \zeta_K^J(b) \big\}
    \end{equation}
    Then \(C\) coaxiomatizes \(\Coalg(F, \Sys)\), i.e., \(\Coalg(F, \Sys) = \Coalg(F; C)\).
    Consequently, the chromatic number of \(\Coalg(F, \Sys)\) is at most \(\card{H(\kappa + \kappa)}\).
\end{restatable}

\begin{proof}
    The proof is easiest to present alongside a diagram.
    Let \((X, \beta)\) be an \(F\)-coalgebra. 
    Since \(K = H(\kappa + \kappa)\) has an EM-algebra structure, for any colouring \(k \colon X \to K\) we obtain an extension \(k^\# \colon HX \to K\), meaning that \(k^\# \circ \eta_{\kappa + \kappa} = k\).
    We also obtain a unique \(\Delta_K \times F\)-coalgebra homomorphism \(\beh \colon (X, \langle \beta, k\rangle) \to (Z_K, \langle \zeta_K, \ell_K\rangle)\) using that \((Z, \langle \zeta_K, \ell_K\rangle)\) is a final \(\Delta_K \times F\)-coalgebra.
    For each \(J\)-path constraint \((\lambda \equiv \rho) \in \Sys\), we obtain the diagram
    \begin{equation}
        \label{eq:coax diagram}
        \begin{gathered}
            \begin{tikzcd}[ampersand replacement = \&]
                X \ar[rr, "\beta^J"] \ar[dd, "\beh"'] \ar[dr, "k"]
                \ar[rrr, to path={ -- ([yshift=2ex]\tikztostart.north) -| node[above, pos=0.25] {\(\eta_X\)} (\tikztotarget)}, rounded corners=5pt]
                \&
                \& JX \ar[r, shift left, "\lambda_X"] \ar[r, shift right, "\rho_X"'] \ar[dd]
                \& HX \ar[dd] \ar[dr, "k^\#"]
                \&
                \\
                \& K \ar[rrr, double, no head, crossing over]
                \&
                \&
                \& K
                \\
                Z_K \ar[rr, "\zeta_K^J"] \ar[ur, "\ell_K"]
                \ar[rrr, to path={ -- ([yshift=-2ex]\tikztostart.south) -| node[below, pos=0.25] {\(\eta_{Z_K}\)} (\tikztotarget)}, rounded corners=5pt]
                \&
                \& JZ_K \ar[r, shift left, "\lambda_{Z_K}"] \ar[r, shift right, "\rho_{Z_K}"']
                \& HZ_{K} \ar[ur, "\ell_K^\#"'] 
                \&
            \end{tikzcd}
        \end{gathered}
    \end{equation} 
    Before we proceed, it will be helpful to have the following fact at our fingertips.
    \medskip

    \noindent{\bf Claim.~}
    For \(\tau \in \{\lambda, \rho\}\), \(k^\# \circ \tau_{X} \circ \beta^J = \ell_K^\# \circ \tau_{Z_K} \circ \zeta_K^J \circ \beh\).
    \medskip

    \noindent\emph{Proof of Claim.}
        Since \(\beh\) is a homomorphism of \(\Delta_K\times F\)-coalgebras, \(\ell_K \circ \beh = k\). 
        By naturality of \(\eta\), \(
            \ell_K^\# \circ H(\beh) \circ \eta_X 
            = \ell_K^\# \circ \eta_{Z_K} \circ \beh 
            = \ell_K \circ \beh
            = k
        \),
        so from uniqueness of \(k^\#\) we know that \(\ell_K^\# \circ H(\beh) = k^\#\). 
        Therefore,
        \begin{align*}
            \ell_K^\# \circ \tau_{Z_K} \circ \zeta_K^J \circ \beh
            &= \ell_K^\# \circ \tau_{Z_K} \circ J(\beh) \circ \beta^J \tag{\cref{lem:always a homom}}\\
            &= \ell_K^\# \circ H(\beh) \circ \tau_X \circ \beta^J \tag{\(\tau\) natural}\\
            &= k^\# \circ \tau_X \circ \beta^J \tag{uniqueness of \(k^\#\)}
        \end{align*}
    This concludes the proof of the claim.\hfill\qed
    \medskip

    \noindent%
    For the main theorem, we argue each direction separately.
    \smallskip

    \noindent%
    (\(\Rightarrow\)) Let us begin by showing that if \((X, \beta) \in \Coalg(F, \Sys)\), then we also have \((X, \beta) \in \Coalg(F;C)\). 
    Given \((\lambda \equiv \rho) \in \Sys\) and a colouring \(k \colon X \to K\), we obtain \eqref{eq:coax diagram}. 
    We need to show that for any \(x \in X\), \(\beh(x) \in C\) (see \eqref{def:coax}). 
    But from the Claim we immediately obtain 
    \begin{align*}
        \ell_K^\# \circ \lambda_{Z_K} \circ \zeta_K^J \circ \beh
        &= k^\# \circ \lambda_{X} \circ \beta^J \\
        &= k^\# \circ \rho_{X} \circ \beta^J \tag{\((X, \beta) \in \Coalg(F, \Sys)\)} \\
        &= \ell_K^\# \circ \rho_{Z_K} \circ \zeta_K^J \circ \beh
    \end{align*}
    It follows that \(\beh(x) \in C\) for all \(x \in X\).
    Since \((\lambda \equiv \rho) \in \Sys\), and the state \(x\) and colouring \(k\) were arbitrary, \((X, \beta)\) satisfies \(C\). \medskip

    \noindent%
    (\(\Leftarrow\)) We actually show the contrapositive: if \((X, \beta) \notin \Coalg(F, \Sys)\), then \((X, \beta) \notin \Coalg(F; C)\).
    Suppose \((X, \beta) \notin \Coalg(F, \Sys)\).
    By \cref{lem:equational}, there is a \((\lambda \equiv \rho) \in \Sys\) and an \(x \in X\) such that \(\lambda_X\circ \beta^J(x) \neq \rho_X \circ \beta^J(x)\). 
    Let \(t = \lambda_X\circ \beta^J(x)\) and \(s = \rho_X\circ \beta^J(x)\).
    Since \(H\) is \(\kappa\)-accessible, there are \(U_s,U_t \subseteq X\) such that \(s \in HU_s\) and \(t \in HU_t\), \(\card{U_s} \le \kappa\) and \(\card{U_t} \le \kappa\).
    Taking their union \(U = U_s \cup U_t\), we have \(\card{U_s \cup U_t} \le \kappa + \kappa\) and \(s,t \in HU \subseteq HX\).

    Now, since \(\card{U} \le \kappa + \kappa\), there is an injective function \(j \colon U \to \kappa + \kappa\). 
    Choose any extension of \(j\) to a function \(j' \colon X \to \kappa + \kappa\), ie., \(j = j' \circ \incl \colon U \subseteq X \to \kappa + \kappa\).
    Composing with the unit of \(H\) we obtain a colouring \(k = \eta_{\kappa + \kappa} \circ j'\) of \(X\),
    \[\begin{tikzcd}[ampersand replacement=\&]
        X \ar[r, blue, "k"] \ar[dr, "j'"'] \& K = H(\kappa + \kappa) \\
        U \ar[u, "\incl"] \ar[r, "j"'] \& \kappa + \kappa \ar[u, "\eta_{\kappa + \kappa}"']
    \end{tikzcd}\]
    and we obtain from the uniqueness of \(k^\#\) that 
    \begin{align*}
		k^\# 
		&= \mu_{\kappa+\kappa} \circ H(k) \tag{def.~\((-)^\#\)} \\
		&= \mu_{\kappa+\kappa} \circ H(\eta_{\kappa + \kappa} \circ j') \tag{def.~\(k\)} \\
		&= \mu_{\kappa+\kappa} \circ H(\eta_{\kappa + \kappa}) \circ H(j') \tag{\(H\) functorial}	\\
		&= H(j')	\tag{monad laws}			
    \end{align*}
    Since \(j'\) is injective on \(U\) and \(H\) preserves monics, \(k^\# = H(j')\) is injective restricted to \(HU \subseteq HX\).
    Concretely, 
    \begin{align*}
        k^\# \circ \incl_{HU} 
        &= H(j') \circ \incl_{HU} \tag{above} \\
        &= H(j') \circ H(\incl_{U}) \tag{\(H\) preserves inclusions} \\
        &= H(j' \circ \incl_U) \tag{\(H\) functorial}\\
        &= H(j) \tag{def.~\(j'\)}
    \end{align*}
    In particular, this means that \(k^\#(t) \neq k^\#(s)\). 
    We arrive at the following calculation: 
    \begin{align*}
        \ell^\# \circ \lambda_{Z_K} \circ \zeta^J \circ \beh(x)  
        &= k^\# \circ \lambda_{X} \circ \beta^J(x)                \tag{Claim} \\
        &= k^\# (t)                                               \tag{def.~\(t\)} \\
        &\neq k^\# (s)                                            \tag{assumed} \\
        &= k^\# \circ \rho_{X} \circ \beta^J(y)                   \tag{def.~\(s\)} \\
        &= \ell^\# \circ \rho_{Z_K} \circ \zeta^J \circ \beh(y)   \tag{Claim}
    \end{align*}
    It follows that \(\beh(x) \notin C\), so \((X, \beta) \notin \Coalg(F;C)\). 
\end{proof}

\begin{exa}
    The identity functor \(\Id \colon \Set \to \Set\) has \(\locCard{\Id} = 1\), and paired with \(\mu = \id{} = \eta\) is a monad.
    It follows that every system of path constraints that are equational over \(\Id\) has a chromatic number of at most \(2\) by \cref{thm:coax}. 
    Concrete examples of path constraints equational over \(\Id\) include those obtained from monoid presentations (\cref{thm:monoid presentations}).
    For example, in the commutativity case, \(K = \Id(1 + 1) = 2\), \(Z_2 = (2 \times B)^{A^*}\), and the behaviour colouring is \(\ell(\theta) = \proj_2(\theta(\varepsilon))\).
    The coequation in \eqref{def:coax} unravels to 
    \begin{align*}
        C 
        &= \{\theta \in (2 \times B)^{A^*} \mid \forall a,b \in A,~\proj_2(\theta(ab)) = \proj_2(\theta(ba))\}
    \end{align*}
    Thus, for a Moore automaton \((X, \beta)\), \((X, \beta)\) satisfies \(C\) if and only if for any colouring \(k \colon X \to 2\), \(x \tr{a} {\cdot} \tr{b} z\) and \(x \tr{b} {\cdot} \tr{a} z'\) implies \(k(z) = k(z')\).
    The latter is equivalent to saying \(z = z'\), so \((X, \beta)\) satisfies \(C\) if and only if \((X, \beta)\) satisfies \(\der a \der b \equiv \der b \der a\).
\end{exa}

\section{The Terminal Net Construction}
\label{sec:terminal net}

In this last technical section, we present an iterative construction that, under moderate conditions on the functor \(F\), produces an \(F\)-coalgebra that is final relative to a given system of singular path constraints \(\Sys\).
The moderate conditions stipulate that \(F\) and the path constraints \(E \in \Sys\) are sufficiently continuous.
The level of continuity required is made precise in \cref{def:pitched} by what are called \emph{pitched limits}, but to reassure the reader in advance, polynomial endofunctors are examples of functors that are continuous enough for our construction to be carried out.
Our construction therefore applies to the case of commutativity for Moore automata in \(\Set\), for example.
Our construction is similar to the terminal sequence constructions of Ad\'amek~\cite{Adamek74} and Barr~\cite{Barr93}, in that it iteratively generates a diagram from a final object and obtains the final coalgebra in the limit. 

For fear of further repetition, we would like to stress that the path constraints considered in this section are not required to be equational, but they are required to be singular, which means that every \(E \in \Sys\) is of the form \(E \subseteq F^n\) for some \(n \in \NN\) (although \(n\) may vary).

To understand the terminal net construction formally, it is necessary to have some familiarity with the \emph{terminal sequence}~\cite{Barr93,Adamek74,Worrell98}.
The terminal sequence is discussed in~\cref{par:terminal} and used in the proof of \cref{thm:converse to covariety}, but a comprehensive treatment can be found, e.g., in~\cite{AdamekMossMilius}.
Roughly, the terminal sequence is a sequence of approximations of a final coalgebra.
Intuitively, one could call the terminal sequence a ``one-dimensional'' approach to constructing a final \(F\)-coalgebra, since every item in the sequence has a unique predecessor.
To incorporate a system of path constraints into the terminal sequence \eqref{eq:terminal sequence}, we need to add to it a ``constraint dimension''.

\begin{restatable}{defi}{definitionofterminalnet}
    \label{def:generalized Adamek-Barr construction}
    Let \(\Sys = \{e^{(i)} \colon E_i \subseteq F^{n_i} \mid i \in I\}\) be a system of singular path constraints indexed by a set \(I\).
    We define the category \(\Word\Sys\) to be the free category generated by the following objects and arrows:
    the objects of \(\Word{\Sys}\) are words \(\Sys_F^* = (\{\code{E}_i \mid i \in I\} \cup \{\code{F}\})^*\) formed from the \(E_i\) in \(\Sys\) and \(F\) treated as formal symbols.
    We write \(\bullet \in \Sys_F^*\) for the empty word.
    The arrows of \(\Word\Sys\) are generated by the following rules:
    \begin{itemize}
        \item there is an arrow \(\mathtt\bang_{w} \colon w \to \bullet\) for each \(w \in \Sys_F^*\),
        \item there is an arrow \(\code{e}_w^{(i)} \colon \code{E}_iw \to \code{F}^{n_i} w\) for each \(i \in I\) and \(w \in \Sys_F^*\), and
        \item for each \(f \colon w \to u\) in \(\Word\Sys\) and \(i \in I\), there are the arrows \(\code{F}(f) \colon \code{F}w \to \code{F}u\) and \(\code{E}_i(f) \colon \code{E}_i w \to \code{E}_i u\).
    \end{itemize}
    Given any word \(w \in \Sys_F^*\), we define \(\WFunc w \colon \Cat \to \Cat\) recursively on \(w\) by \(\WFunc \bullet = \Id\), \(\WFunc {\code{F}w} = F\circ \WFunc w\), and \(\WFunc {\code{E}_iw} = E_i\circ \WFunc w\).
    The \emph{terminal net (for \(\Sys\))} is the diagram \(\Net \colon \Word\Sys \to \Cat\) defined on objects by 
    \(\Net w = \WFunc{w}1\) for any \(w \in \Word\Sys\).
    On arrows,
    \begin{align*}
        \Net(\bang_w) &= \bang_{\Net w} \colon \Net w \to 1
        & 
        \Net(\code{e}_w^{(i)}) &= e_{\Net w}^{(i)} \colon E_i\Net w \to F^{n_i}\Net w
        \\
        \Net(\id{w}) &= \id{\Net w} \colon \Net w \to \Net w
        & 
        \Net(\code{F}(f)) &= F\Net(f) \colon \Net \code{F}w \to \Net\code{F}u
        \\
        \Net(g \circ f) &= \Net(g) \circ \Net(f) \colon \Net w \to \Net v
        &
        \Net(\code{E}_i(f)) &= E_i \Net(f) \colon \Net \code{E}_iw \to \Net\code{E}_iu
    \end{align*}
    where \(w,u,v\in \Sys_F^*\), \(i \in I\), and \(g \colon u \to v\) and \(f \colon w \to u\) in \(\Word\Sys\).
\end{restatable}

\noindent%
An artist's depiction of the terminal net can be found in \cref{fig:singular net}, where we see the first few layers of the terminal net of a single singular path constraint of length \(2\).

\begin{figure}
    \[\begin{tikzcd}[row sep=1em]
        1 & F1 & {F^21} & {F^31} && {F^41} & \cdots \\
        & E1 \\
        && EF1 & FE1 \\
        &&& FEF1 & {F^2E1} & EE1 & {EF^21}
        \arrow["{\bang_{\code F}}"', from=1-2, to=1-1]
        \arrow["{F(\bang_{\code F})}"', from=1-3, to=1-2]
        \arrow["{F^2(\bang_{\code F})}"', from=1-4, to=1-3]
        \arrow["{F^3(\bang_{\code F})}"', from=1-6, to=1-4]
        \arrow[from=1-7, to=1-6]
        \arrow["{\bang_{\code E}}", from=2-2, to=1-1]
        \arrow["{e_1}"', hook, from=2-2, to=1-3]
        \arrow["{e_{F1}}", hook, from=3-3, to=1-4]
        \arrow["{E(\bang_{\code F})}", from=3-3, to=2-2]
        \arrow["{F(e_1)}"', hook, from=3-4, to=1-4]
        \arrow["{F(e_{F1})}", from=4-4, to=1-6]
        \arrow["{FE(\bang_{\code F})}", from=4-4, to=3-4]
        \arrow["{F^2(e_1)}"', from=4-5, to=1-6]
        \arrow["{e_{E1}}", from=4-6, to=4-5]
        \arrow["{E(e_1)}"', from=4-6, to=4-7]
        \arrow["{e_{F^21}}"', hook', from=4-7, to=1-6]
    \end{tikzcd}\]
    \caption{
        \label{fig:singular net}
        The first few layers of the terminal net drawn for a single singlular path constraint of length \(2\), \(e \colon E \subseteq F^2\).
        The terminal sequence appears in the top row, and each subsequent row corresponds to words of increasing length.
    }
\end{figure}

Write \(w\Word\Sys\) for the full subcategory of \(\Word\Sys\) consisting of words that begin with \(w\), and write \(\Incl_w\colon w\Word\Sys \hookrightarrow \Word\Sys\) for the inclusion functor.
Since \(\Cat\) is complete, for each \(w \in \Sys_F^*\) we obtain the following two cones and cone homomorphism:
\begin{enumerate}
    \item\label{item:limitWT} a limit cone \((\lim (\WFunc w \circ \Net), c^w)\) for the functor \(\WFunc w \circ \Net\),
    
    \item\label{item:limitTIW} a limit cone \((\lim (\Net \circ \Incl_w), d^w)\) for the functor \(\Net \circ \Incl_w\),
    
    \item\label{item:conerest} a cone homomorphism \(b^w \colon (\lim \Net , c^\bullet) \to (\lim (\Net \circ \Incl_w), d^w)\).
\end{enumerate} 
It is worth noting that \((\lim \Net , c^\bullet) = (\lim \Net , d^{\bullet})\), so the choice of \(c^\bullet\) instead of \(d^\bullet\) in \eqref{item:conerest} is cosmetic.
The key to the terminal net construction producing a final coalgebra (\cref{thm:generalized Adamek-Barr construction}) is that \eqref{item:limitWT} and \eqref{item:limitTIW} coincide when \(F\) and every \(E_i \in \Sys\) preserve limits of \emph{pitched diagrams}, which we define as follows.

\begin{defi}\label{def:pitched}
    A \emph{pitched diagram} is a commutative diagram \(D \colon \mathbf D \to \Cat\) containing a \emph{central} \(\omega^{\op}\)-chain \(\{\alpha_i \colon A_{i+1} \to A_i\}_{i \in \NN} \subseteq \mathbf D\), meaning that every object \(B\) of \(\mathbf D\) has an arrow \(e_B \colon B \to A_i\) for some \(i \in \NN\) (its \emph{pitch}) such that \(D(e_B)\) is monic.
    A functor is \emph{pitched-continuous} if it preserves limits of pitched diagrams.
\end{defi}

The terminal net is of this form: it is commutative by naturality of each \(e \colon E \subseteq F^n\) in \(\Sys\), and its central \(\omega^\op\)-chain is \(\{\code{F}^j(\bang_{\code{F}}) \colon \code{F}^{j+1} \to \code{F}^j\}_{j \in \NN}\).
Since \(F\) and every \(E \in \Sys\) preserves monics, by induction on \(w \in \Sys_F^*\), one can always find an \(n \in \NN\) and an arrow \(w \to \code{F}^n\) in \(\Word{\Sys}\) such that \(\Net(w)\) is monic.

Generally, for any \(w \in \Sys_F^*\), \(\{w\code{F}^j(\bang_{w\code{F}}) \colon w\code{F}^{j+1} \to w\code{F}^j\}_{j \in \NN}\) is a central \(\omega^\op\)-chain for \(w\Word{\Sys}\), so \(\Net \circ \Incl_w\) is also pitched.
It is furthermore true that for any functor \(G \colon \Cat \to \Cat'\) that preserves monics, if \(D \colon \mathbf D \to \Cat\) is pitched, then so is \(G \circ D\).
This in particular holds for \(\WFunc w\), so \(\WFunc w \circ \Net\) is pitched for all \(w \in \Sys_F^*\).

\begin{restatable}{lem}{inclusionlimit}
    \label{lem:W inclusion limit}
    For any \(w \in \Sys_F^*\), there is a unique cone homomorphism \[
        \phi^w \colon (\lim (\Net \circ \Incl_w), d^w) \to (\lim (\WFunc w \circ \Net), c^w)
    \]
	If \(F\) and every \(E_i \in \Sys\) is pitched-continuous, then \(\phi^w\) is a cone isomorphism.
\end{restatable}

\begin{proof}
    For the first statement, it suffices to see that \((\lim (\Net \circ \Incl_w), d^w)\) is a cone for \(\WFunc w \circ \Net\), as this would induce a unique cone homomorphism \(\phi^w \colon \lim(\Net \circ \Incl_w) \to \lim(\WFunc w \circ \Net)\).                                                           
    For the second, it suffices to show that \((\lim (\WFunc w \circ \Net), c^w)\) is a cone for \(\Net \circ \Incl_w\) using pitched continuity of each \(\WFunc w\), since this would induce an inverse for \(\phi^w\).                                                                                            

    For the first statement, let \(g \colon u \to v\) in \(\Word\Sys\).
    We need to show \(\WFunc w\Net (g) \circ d_{wu}^w = d_{wv}^w\).
    But this is immediate from the equation \(\WFunc w\Net (g) = \Net(w(g))\) (shown by an easy induction on \(w\)), because \(w(g) \colon wu \to wv\) and \((\lim(\Net \circ \Incl_w), d^w)\) is a cone for \(\Net \circ \Incl_w\).                                                                     
    In a diagram,
    \[
    \begin{tikzcd}[column sep=10pt]
        & \lim(\Net \circ \Incl_w)
            \ar[dl, "d_{wu}^w"']
            \ar[dr, "d_{wv}^w"]
        &
        \\
        \Net \Incl_w wu
            \ar[rr, "\Net \Incl_w(w(g))"]
            \ar[d, no head, double]
        && \Net \Incl_w wv
            \ar[d, no head, double]
        \\                
        \WFunc w\Net u    
            \ar[rr, "\WFunc w\Net (g)"]                                                                         
        && \WFunc w\Net v 
    \end{tikzcd}          
    \]                    
    The triangle in the top of the diagram above commutes because \((\lim(\Net \circ \Incl_w), d^w)\) is a cone.
    The square commutes by definition of \(\Net\).                                                              
    Thus, \(d^w \colon \Delta_{\lim(\Net \circ \Incl_w)} \Rightarrow \WFunc{w} \circ \Net\) is a natural transformation, i.e., a cone for \(\WFunc{w} \circ \Net\).
    This induces our cone homomorphism \(\phi^w \colon (\lim(\Net \circ \Incl_w), d^w) \to (\lim(\WFunc w \circ \Net), c^w)\).     

    Let us move on to the second statement.                                                                     
    Let \(w, u, v \in \Sys_F^*\) and consider an arbitrary arrow \(g \colon wu \to wv\) in \(w\Word\Sys\).      
    We show by induction on \(g\) that                                                                          
    \begin{equation}                   
        \Net(g) \circ c^w_{u} = c^w_{v}                                                                         
    \end{equation}                    
    where \((\lim(\WFunc{w} \circ \Net), c^w)\) is a limit cone for \(\WFunc{w} \circ \Net\).                   
    \begin{description}   
        \item[Base Cases] 
        In the base case, either \(g = \id{wu}\), \(g = \bang_{wu}\), or \(w = \bullet\) and \(g = \code{e}^{(i)}_{u}\).                                           
        The first situation is trivial because \(\Net(\id{wu}) = \id{\Net wu}\), and \(wu = wv\) implies \(u = v\).                                                
        In the second and third situation, we necessarily have \(w = \bullet\).                                 
        This is because \(wv = \bullet\) implies \(w = v = \bullet\) in the second situation, and in the third, we assume it!                                      
        But then, \[      
            \WFunc w \circ \Net                                                                                 
            = \WFunc \bullet \circ \Net                                                                         
            = \Net        
            = \Net \circ \Incl_\bullet                                                                          
            = \Net \circ \Incl_w                                                                                
        \]                
        and our cones for \(\WFunc w \circ \Net\) and \(\Net \circ \Incl_w\) coincide.                          
        That is, we would have \(c_{u}^\bullet = d_{\bullet u}^\bullet\) for any \(u \in \Sys_F^*\).            
        Thus,             
        \[                
            \Net(g) \circ c_u^w                                                                                 
            = \Net(g) \circ d_{\bullet u}^\bullet                                                               
            = d_{\bullet v}^\bullet                                                                             
            = c_{v}^\bullet                                                                                     
            = c_{v}^w     
        \]                
                          
        \item[Induction Step 1] (Composition of arrows in \(\Word\Sys\)) For \(f \colon wu \to wv\) and \(g \colon wv \to wz\), and assuming                       
        \(                
            \Net(f) \circ c^w_{u} = c^w_{v}                                                                     
        \)                
        and               
        \(                
            \Net(g) \circ c^w_{v} = c^w_{z}                                                                     
        \),               
        we compute        
        \begin{equation*} 
            \Net(g \circ f) \circ c^w_{u}                                                                       
            = \Net(g) \circ \Net(f) \circ c^w_{u}                                                               
            = \Net(g) \circ c^w_v                                                                               
            = c^w_{z}     
        \end{equation*}   
        as desired.       
        This handles the composition case.                                                                      
                          
        \item[Induction Step 2] (Application of \(E_i\) or \(F\))                                               
        Where \(w = \code{G}w'\) for \(w' \in \Sys_F^*\), let \(f \colon w'u \to w'v\) and \(g = \code{G}(f) \colon wu \to wv\) for either \(\code{G} = \code{F}\) or \(\code{G} = \code{E}_i\) for some \(i \in I\), and assume that \(\Net(f) \circ c_u^w = c_v^w\).                                               
        Since \(G := \WFunc {\code{G}}\) is pitched-continuous and \(\WFunc{w'} \circ \Net\) is a pitched diagram (this is true for any monic-preserving functor in place of \(\WFunc{w'}\)), there is an isomorphism of cones                                                                                       
        \[                
            \tau \colon   
            (\lim(\WFunc {w} \circ \Net), c^{w})                                                                
            = (\lim(G \circ \WFunc {w'} \circ \Net), c^{w})                                                     
            \xrightarrow{\cong} (G(\lim(\WFunc {w'} \circ \Net)), G(c^{w'}))                                    
        \]                
        We obtain the following diagram, in which all three inner triangles commute:                            
        \begin{equation*}\begin{tikzcd}                                                                         
            \Net (wu) = G\Net(w'u)                                                                              
                \ar[rr, to path={ -- ([yshift=2ex]\tikztostart.north) -| node[above, pos=0.25] {\footnotesize\(\Net(g) = \Net(\code{G}(f)) = G\Net(f)\)} (\tikztotarget)}, rounded corners=5pt]                                                                                                                      
            & G(\lim(\WFunc {u} \circ \Net))                                                                    
                \ar[l, "G(c_{u}^{w'})"']                                                                        
                \ar[r, "G(c_v^{w'})"]                                                                           
            & \Net(wv) = G\Net(w'v)                                                                             
            \\            
            & \lim(\WFunc {\code{G}u} \circ \Net)                                                               
                \ar[ul, out=175, in=-35, "c_{u}^w"]                                                             
                \ar[ur, out=5, in=-145, "c_v^w"']                                                               
                \ar[u, "\tau"]                                                                                  
        \end{tikzcd}\end{equation*}                                                                             
        Indeed, the top triangle commutes because                                                               
        \[                
            G\Net(f) \circ G(c_u^{w'}) = G(\Net(f) \circ c_u^{w'}) = G(c_v^{w'})                                
        \]                
        by the induction hypothesis.                                                                            
        Thus, the outer triangle commutes as well.                                                              
    \end{description}     
                          
    \noindent This concludes the induction steps, as the only recursive steps in the construction of arrows in \(\Word\Sys\) are the applications of \(\code{E}_i\) and \(\code{F}\) to and composition of arrows.                                                                                                   
    It follows that \((\lim(\WFunc{w} \circ \Net), c^w)\) is a cone for \(\Net \circ \Incl_w\) as well as \(\WFunc{w} \circ \Net\).                           
    This induces a unique cone homomorphism \((\phi^{w})^{-1} \colon (\lim(\WFunc{w} \circ \Net), c^w) \to (\lim(\Net \circ \Incl_w), d^w)\), which is necessarily inverse to \(\phi^w\) by the universal properties of the limiting cones for \(\WFunc{w} \circ \Net\) and \(\Net \circ \Incl_w\).                  
\end{proof}

To see the significance of \cref{lem:W inclusion limit}, observe that, intuitively, both \(\sem{w} \circ \Net\) and \(\Net \circ \Incl_w\) can be seen as inclusion functors.
Since \(\sem{w}\Net u = \Net wu\) for any \(u \in \Sys_F^*\), \(\sem{w} \circ \Net\) represents the inclusion of the subcategory of \(\Cat\) consisting of all objects \(\sem{wu}1\) and arrows of the form \(\sem{w}(\Net(f))\) for \(f \colon u \to u'\).
In contrast, \(\Net \circ \Incl_w\) is the inclusion of the full subcategory \(w\Word\Sys\) of \(\Word\Sys\) followed by \(\Net\).

Note that for any \(w,u \in \Sys_F^*\), \(\Net \Incl_w(wu) = \Net (wu) = \WFunc w \Net u\), so the codomain of \(d_{wu}^{w}\) coincides with the codomain of \(c_u^{w}\).

\begin{restatable}[Terminal Net Construction]{thm}{adamekbarrthm}%
    \label{thm:generalized Adamek-Barr construction}
    Let \(F \colon \Cat \to \Cat\) and let \(\Sys\) be a system of singular path constraints such that \(F\) and every \(E \in \Sys\) is pitched-continuous. 
    Let \(\Net \colon \Word\Sys \to \Cat\) be the terminal net for \(\Sys\), and (using pitched continuity) let \(\tau \colon (\lim (F \circ \Net), c^{\code{F}}) \to (F(\lim \Net ), F(c^\bullet))\) be the assumed cone isomorphism.
    Let \(b^{\code F}\) be the cone homomorphism induced by the restriction of the limit cone for \(\Net\) to the subcategory \({\code F}\Word\Sys\), and \(\phi^{\code F}\) be the cone homomorphism constructed in \cref{lem:W inclusion limit}.
	Then \((Z_\Sys, \zeta_\Sys)\) is a final \(F\)-coalgebra relative to \(\Sys\), where 
    
	\begin{gather}
		\label{eq:final relative coalg}
        \begin{gathered}
            \begin{tikzcd}[ampersand replacement = \&]
                Z_\Sys = \lim \Net 
                    \ar[d, "b^{\code F}"]
                    \ar[r, blue, "\zeta_\Sys"]
                \& F(\lim \Net)
                \\
                \lim(\Net \circ \Incl_{\code F})
                    \ar[r, "\phi^{\code F}"]
                \& \lim(F \circ \Net)
                    \ar[u, "\tau^{\code F}"]
            \end{tikzcd}
        \end{gathered}
        \qquad 
        \mathcolor{blue}{\zeta_\Sys} = \tau^{\code F} \circ \phi^{\code{F}} \circ b^{\code{F}}
    \end{gather}
\end{restatable}

\begin{proof}                                                                                 
    Let us start by showing that \((\lim \Net, \zeta_\Sys)\) satisfies \(\Sys\).              
        
    Write \(\Sys = \{e^{(i)} \colon E_i \subseteq F^{n_i} \mid i \in I\}\), and recall the recursive definition \(\zeta_{\Sys}^0 = \id{\lim \Net}\) and \(\zeta_{\Sys}^{n+1} = F(\zeta_{\Sys}^n) \circ \zeta_{\Sys}\).                                                                             
    We need to factor \(\zeta_{\Sys}^{n_i}\) through \(e^{(i)} \colon E_i \to F^{n_i}\) for each \(i \in I\).                                    
    To this end, we argue that \(\zeta_{\Sys}^k\) is a homomorphism of \(F^k \circ \Net\) cones, \((\lim \Net, c_{\code F^k}^\bullet) \to (F^k(\lim \Net), F^k(c^\bullet))\) for all \(k \ge 0\).                                                                                                  
    We proceed by induction on \(k\).                                                         
        
    \begin{description}                                                                       
        \item[Base Case] For \(k = 0\), \(\zeta_{\Sys}^k = \id{\lim \Net}\colon (\lim \Net, c^\bullet) \to (F^0(\lim \Net), F^0(c^\bullet))\) is trivially a cone homomorphism.                                                                                                                    
        
        \item[Induction step] Suppose we have established that                                
        \begin{gather*}                                                                       
            \zeta_{\Sys}^k \colon (\lim \Net, c^\bullet) \to (F^k(\lim \Net), F^k(c^\bullet)) 
            \\                                                                                
            \begin{tikzcd}[ampersand replacement=\&]                                          
                \lim \Net                                                                     
                    \ar[dr, "c_{\code F^k}^\bullet"']                                         
                    \ar[rr, "\zeta_\Sys^k"]                                                   
                \&\& F^k(\lim \Net)                                                           
                    \ar[dl, "F^k(c^\bullet)"]                                                 
                \\                                                                            
                \& F^k \circ \Net                                                             
            \end{tikzcd}                                                                      
        \end{gather*}                                                                         
        is a cone homomorphism.                                                               
        Then                                                                                  
        \begin{gather*}                                                                       
            F(\zeta_{\Sys}^k) \colon (F(\lim \Net), F(c_{\code F^k}^\bullet)) \to (F^k(\lim \Net), F^{k+1}(c^\bullet))                           
            \\                                                                                
            \begin{tikzcd}[ampersand replacement=\&]                                          
                F(\lim \Net)                                                                  
                    \ar[dr, "F(c_{\code F^k}^\bullet)"']                                      
                    \ar[rr, "F(\zeta_{\Sys}^k)"]                                              
                \&\& F^{k+1}(\lim \Net)                                                       
                    \ar[dl, "F^{k+1}(c^\bullet)"]                                             
                \\                                                                            
                \& F^{k+1} \circ \Net                                                         
            \end{tikzcd}                                                                      
        \end{gather*}                                                                         
        is also a cone homomorphism by functoriality.                                         
        By definition, \(\zeta_{\Sys} = \tau^{\code F} \circ \phi^{\code{F}} \circ b^{\code{F}}\) is a composition of cone homomorphisms and as is therefore a cone homomorphism.                                                                                                                  
        It follows that \(\zeta_{\Sys}^{k+1} = F(\zeta_{\Sys}^k) \circ \zeta_{\Sys}\) is also a cone homomorphism,                               
        \begin{equation*}                                                                     
            \begin{tikzcd}[ampersand replacement=\&]                                          
                \&  (F^{k+1}(\lim \Net), F^{k+1}(c^\bullet))                                  
                \\                                                                            
                (\lim \Net, c_{\code F^{k+1}}^\bullet)                                        
                    \ar[ur, "\zeta^{k+1}"]                                                    
                    \ar[r, "\zeta_\Sys"']                                                     
                \& (F(\lim \Net), F(c_{\code F^k}^\bullet))                                   
                    \ar[u, "F(\zeta_\Sys^k)"'] \&                                             
            \end{tikzcd}                                                                      
        \end{equation*}                                                                       
    \end{description}                                                                         
    This concludes the proof that \(\zeta_{\Sys}^k\) is a cone homomorphism for all \(k \ge 0\).                                                 
    In particular, \(\zeta_\Sys^{n_i} \colon (\lim \Net, c_{\code F^{n_i}}^\bullet) \to (F^{n_i}(\lim \Net), F^{n_i}(c^\bullet))\) is a cone homomorphism for each \(i \in I\).                                                                                                       
        
    Now fix an arbitrary \(i \in I\).                                                         
    Since \(F\) preserves limits of pitched diagrams, the composition \(F^{n_i}\) preserves limits of pitched diagrams.                          
    This tells us that \((F^{n_i}(\lim \Net), F^{n_i}(c^\bullet))\) is a limit cone for \(F^{n_i} \circ \Net\), and therefore \(\zeta_{\Sys}^{n_i}\) is the \emph{unique} cone homomorphism \((\lim \Net, c_{\code F^{n_i}}^\bullet) \to (F^{n_i}(\lim \Net), F^{n_i}(c^\bullet))\) for \(F^{n_i}\).
    Now, since \(e^{(i)} \colon E_i \subseteq F^{n_i}\) is natural, the arrow                 
    \[  
        e_{\lim \Net}^{(i)} \colon (E_i(\lim \Net), E_i(c_{\code F^{n_i}}^{\bullet})) \to (F^{n_i}(\lim \Net), F^{n_i}(c^{\bullet}))             
    \]  
    is a cone homomorphism.                                                                   
    And again, it is therefore the unique cone homomorphism, because we have assumed that \(E_i\) preserves limits of pitched diagrams.          
    Thus, to see that \(\zeta_\Sys^{n_i}\) factors through the arrow \(e_{\lim \Net}^{(i)}\), it suffices to construct any cone homomorphism \((\lim \Net, c_{\code F^{n_i}}^{\bullet}) \to (F^{n_i}(\lim \Net), F^{n_i}(c^\bullet))\) that happens to factor through \(E_i(\lim \Net)\).          
    We can describe this arrow explicitly: it is                                              
    \[  
        \psi                                                                                  
        = \mathcolor{blue}{e_{\lim \Net}^{(i)}}                                                 
        \circ \mathcolor{blue}{\tau^{\code E_i} }                                             
        \circ \mathcolor{blue}{\phi^{\code E_i} }                                             
        \circ \mathcolor{blue}{b^{\code E_i}}                                                 
    \]  
    To see that \(\psi\) is a cone homomorphism, let \(w \in \Sys_F^*\).                      
    We need to show that \(F^{n_i}(c_w^\bullet) \circ \psi = c^{\bullet}\).                   
    To that end, consider the diagram in \eqref{eq:psi}.                                      
    \begin{equation}                                                                          
        \label{eq:psi}                                                                        
        \begin{tikzcd}[                                                                       
                column sep=5em,                                                               
                row sep=3em,                                                                  
                every label/.append style = {font = \small}                                   
            ]                                                                                 
            \lim \Net                                                                         
                \ar[rr, "c_{\code F^{n_i}w}^\bullet"]                                         
                \ar[d, blue, "b^{\code E_i}"]                                                 
                \ar[dr, "c_{\code{E}_iw}^\bullet"]                                            
            && \Net \code F^{n_i} w = F^{n_i}\Net w                                           
            \\                                                                                
            \lim(\Net \circ \Incl_{\code E_i})                                                
                \ar[d, blue, "\phi^{\code E_i}"]                                              
                \ar[r, "d_{\code E_iw}^{\code E_i}"]                                          
            & \Net \code E_i w = E_i\Net w                                                    
                \ar[ur, "e_{\Net w}^{(i)}"]                                                   
            &                                                                                 
            \\                                                                                
            \lim(E_i \circ \Net)                                                              
                \ar[r, blue, "\tau^{\code E_i}"]                                              
                \ar[ur, "c_{w}^{\code E_i}"]                                                  
            & E_i(\lim \Net)                                                                  
                \ar[u, "E_i(c_w^\bullet)"']                                                   
                \ar[r, blue, "e_{\lim \Net}^{(i)}"]                                           
            & F^{n_i}(\lim \Net)                                                              
                \ar[uu, "F^{n_i}(c_w^\bullet)"]                                               
        \end{tikzcd}                                                                          
    \end{equation}                                                                            
    In the diagram above, the entire left side and bottom side make up \(\psi\).              
    All of the triangles on the left commute because the outer arrows are cone homomorphisms, and the triangle on top commutes because \((\lim \Net, c^\bullet)\) is a cone for \(\Net\).                                                                                                          
    The square commutes because \(e^{(i)}\) is natural.
    This shows that \(\psi\) is a cone homomorphism, so \(\zeta_{\Sys}^{n_i} = \psi\) follows from uniqueness.

    It remains to see that \((\lim \Net, \zeta_{\Sys})\) is also final relative to \(\Sys\).
    To this end, consider an arbitrary \(F\)-coalgebra \((X, \beta)\) that satisfies \(\Sys\), and let \(\beta^{n_i} = e_X^{(i)} \circ \varepsilon_i\) be the factorization of \(\beta^{n_i}\) through \(e_X^{(i)} \colon E_iX \hookrightarrow F^{n_i}X\) for each \(i \in I\).                    
    We are going to construct a coalgebra homomorphism \((X, \beta) \to (\lim \Net, \zeta_\Sys)\) and show that it is unique.                    
        
    To begin with, let us define an arrow \(\beta^w \colon X \to \WFunc wX\) for each \(w \in \Sys_F^*\) by recursion on \(w\).                  
    In the base case, \(w = \bullet\) and \(\WFunc w = \Id\), and so we simply take \(        
        \beta^\bullet = \id{X}                                                                
    \). 
    For the recursion steps, define \(                                                        
        \beta^{\code E_i w} = E_i(\beta^w) \circ \e_i                                         
    \) and \(                                                                                 
        \beta^{\code Fw} = F(\beta^w) \circ \beta                                             
    \). 
    Observe the following two facts, each of which can be proven by induction:                
    \begin{enumerate}                                                                         
        \item \label{it:fact1} \(\beta^{k} = \beta^{\code F^{k}}\) for all \(k \in \NN\)  
        \item \label{it:fact2} for any \(w,u \in \Sys_F^*\), \(\WFunc w(\beta^{u}) \circ \beta^w = \beta^{wu}\)                                 
    \end{enumerate}                                                                           
    We are now interested in the family of maps given by \(\WFunc w(\bang_{X}) \circ \beta^w \colon X \to \WFunc w 1 = \Net w\) for each \(w \in \Sys_F^*\), because, as we will see shortly, \((X, \lambda)\) is a cone for \(\Net\), where                                                       
    \begin{equation}\label{eq:lambda} 
        \lambda_w = \WFunc w(\bang_X) \circ \beta^w
    \end{equation}            
    Indeed, we are about to show that for any \(w,u \in \Sys_F^*\) and \(g \colon w \to u\) in \(\Word\Sys\),
    \[                        
        \Net(g) \circ \lambda_w       
        = \Net(g) \circ \WFunc w(\bang_{X}) \circ \beta^w                                                    
        \stackrel{\mathcolor{blue}*}{=} \WFunc u(\bang_X) \circ \beta^u                                      
        = \lambda_u           
    \]                        
    The first and last equality are the definition.
    We show \(\mathcolor{blue}*\) by induction on \(g\).                                                     
    \begin{description}       
        \item[Base Case 1] The base case in which \(g = \id{w}\) is trivial.                                 
                              
        \item[Base Case 2] Suppose \(g = \bang_w\).
        In this case, \(u = \bullet\), and we have 
        \begin{equation*}     
            \Net(g) \circ \WFunc w(\bang_{X}) \circ \beta^w                                                  
            = \Net(\bang_w) \circ \WFunc w(\bang_{X}) \circ \beta^w                                          
            = \bang_{\WFunc w1} \circ \beta^w    
            = \bang_{X}       
            = \WFunc \bullet(\bang_X) \circ \beta^\bullet                                                    
        \end{equation*}       
        as desired, since \(u = \bullet\).       
                              
        \item[Base Case 3] (\(g = \code e_{v}^{(i)}\)) Here, \(w = \code E_iv\) and \(u = \code F^{n_i}v\) for some \(v\).                  
        We have               
        \begin{align*}        
            \Net(g) \circ \WFunc w(\bang_{X}) \circ \beta^w                                                  
            &= e_{\Net v}^{(i)} \circ E_i \WFunc v(\bang_{X}) \circ \beta^w                                  
                \tag{def.~\(g\), \(\Net(\code e_v^{(i)}) = e_{\Net v}^{(i)}\)}                               
            \\                
            &= F^{n_i}\WFunc v(\bang_{X}) \circ e_{\WFunc vX}^{(i)} \circ \beta^w                            
                \tag{\(e^{(i)}\) natural}        
            \\                
            &= F^{n_i} \WFunc v(\bang_{X}) \circ e_{\WFunc vX}^{(i)} \circ \beta^{\code E_i v}               
                \tag{\(w = \code E_iv\)}         
            \\                
            &= F^{n_i} \WFunc v(\bang_{X}) \circ e_{\WFunc vX}^{(i)} \circ E_i(\beta^{v}) \circ \e_i         
                \tag{def.~of \(\beta^{\code E_i v}\)}                                                        
            \\                
            &= F^{n_i}\WFunc v(\bang_{X}) \circ F^{n_i}(\beta^{v}) \circ e_{X}^{(i)} \circ \e_i              
                \tag{\(e^{(i)}\) natural}        
            \\                
            &= F^{n_i}\WFunc v(\bang_{X}) \circ F^{n_i}(\beta^{v}) \circ \beta^{\code F^{n_i}}               
                \tag{\(\beta^{n_i} = e_X^{(i)} \circ \e_i\)}                                                 
            \\                
            &= F^{n_i}\WFunc v(\bang_{X}) \circ \beta^{\code F^{n_i}v}                                       
                \tag{\cref{it:fact2} above}      
            \\                
            &= \WFunc u(\bang_{X}) \circ \beta^{u}  \tag{\(u = \code F^{n_i}v\)}                             
        \end{align*}          
                              
        \item[Induction Step 1] (Composition of arrows in \(\Word\Sys\))                                     
        Suppose \(g = g_1 \circ g_0\), where \(g_0 \colon w \to u\), \(g_1 \colon u \to v\),                 
        and for our induction hypothesis we are assuming that                                                
        \[                    
            \Net(g_0) \circ \WFunc w(\bang_{X}) \circ \beta^w = \WFunc v(\bang_X) \circ \beta^v              
        \]                    
        as well as            
        \[                    
            \Net(g_1) \circ \WFunc v(\bang_{X}) \circ \beta^v = \WFunc u(\bang_X) \circ \beta^u              
        \]                    
        Then we compute as follows:   
        \begin{align*}        
            &\Net(g) \circ \WFunc w(\bang_{X}) \circ \beta^w  \\                                             
            &= \Net(g_1 \circ g_0) \circ \WFunc w(\bang_{X}) \circ \beta^w \tag{def.~\(g\)}                  
            \\                
            &= \Net(g_1) \circ \Net(g_0) \circ \WFunc w(\bang_{X}) \circ \beta^w                             
            \tag{\(\Net\) functorial} \\         
            &= \Net(g_1) \circ \WFunc v(\bang_{X}) \circ \beta^v                                             
            \tag{ind.~hyp.} \\
            &= \WFunc u(\bang_X) \circ \beta^u   
            \tag{ind.~hyp.}   
        \end{align*}          
                              
        \item[Induction Step 2] (Application of \(E_i\))                                                     
        Suppose that \(       
            w = \code{E_i}w'  
        \), \(                
            u = \code E_i u'  
        \), and \(            
            g = \code E_i(h)  
        \) for some \(h \colon w' \to u'\) such that                                                         
        \(                    
            \Net(h) \circ \WFunc {w'}(\bang_{X}) \circ \beta^{w'}                                            
            = \WFunc {u'}(\bang_X) \circ \beta^{u'}
        \).                   
        Then                  
        \begin{align*}        
            &\Net(g) \circ \WFunc w(\bang_{X}) \circ \beta^w \\                                              
            &= E_i\Net(h) \circ \WFunc {\code E_iw'}(\bang_X) \circ \beta^{\code E_iw'}                      
                \tag{\(\Net (\code E_i(h)) = E_i\Net(h)\)}                                                   
            \\                
            &= E_i\Net(h) \circ \WFunc {\code E_iw'}(\bang_X) \circ E_i(\beta^{w'}) \circ \e_i               
                \tag{def.~\(\beta^{\code E_iw'}\)} 
            \\                
            &= E_i\Net(h) \circ E_i\WFunc {w'}(\bang_X) \circ E_i(\beta^{w'}) \circ \e_i                     
                \tag{def.~\(\WFunc {\code E_iw'}\)}
            \\                
            &= E_i(h \circ \WFunc {w'}(\bang_X) \circ \beta^{w'}) \circ \e_i                                 
                \tag{\(E_i\) a functor}          
            \\                
            &= E_i(\WFunc {u'}(\bang_X) \circ \beta^{u'}) \circ \e_i                                         
                \tag{ind.~hyp.}       
            \\                
            &= E_i\WFunc {u'}(\bang_X) \circ E_i(\beta^{u'}) \circ \e_i                                      
                \tag{\(E_i\) a functor} \\       
            &= \WFunc u(\bang_X) \circ \beta^{u} \tag{def.~\(u\), \(\beta^u\)}                               
        \end{align*}          
                              
        \item[Induction Step 3] (Application of \(F\))                                                       
        Similar to Induction step 2.  
    \end{description}         
                              
    This concludes the proof that \((X, \lambda)\) is a cone for \(\Net\).                                   
    Thus, we obtain a unique cone homomorphism \(\beh_\beta \colon (X, \lambda) \to (\lim \Net, c^\bullet)\).
    To see that \(\beh_\beta\) is an \(F\)-coalgebra homomorphism, observe that there is also the cone homomorphism \(\zeta_{\Sys} \circ \beh_\beta \colon (X, \lambda) \to (F(\lim \Net), F(c^\bullet))\) for \(\Net\), which is unique because \(\Net\) is pitched and \(F\) is pitched-continuous.
    By uniqueness, it therefore suffices to show that the composition \(F(\beh_\beta) \circ \beta\) is a cone homomorphism of the same type.
    To this end, let \(w \in \Sys_F^*\), and calculate                                                       
    \begin{align*}            
        F(c_w^\bullet) \circ F(\beh_\beta) \circ \beta                                                       
        &= F(c_w^\bullet \circ \beh_\beta) \circ \beta  \tag{\(F\) a functor}\\                              
        &= F(\WFunc w(\bang_X) \circ \beta^w) \circ \beta  \tag{\(\beh_\beta\) a cone homom.} \\             
        &= F\WFunc w(\bang_X) \circ F(\beta^w) \circ \beta  \tag{\(F\) a functor}\\                          
        &= \WFunc {\code Fw}(\bang_X) \circ \beta^{\code Fw} \tag{def.~of \(\beta^{\code Fw}\)} \\           
        &= \lambda_w \tag{def.~\(\lambda\)}      
    \end{align*}              
    as desired.               
    Again, it follows from \(\Net\) being pitched and the pitched continuity of \(F\) that there is at most one cone homomorphism \((X, \lambda) \to (F(\lim \Net), F(c^\bullet))\) for \(\Net\), so we must have                                                                                  
    \[                        
        \zeta \circ \beh_\beta = F(\beh_\beta) \circ \beta                                                   
    \]                        
    Hence, \(\beh_\beta\) is a coalgebra homomorphism from an \(F\)-coalgebra that satisfies \(\Sys\).
                              
    So far, we have seen that \((\lim \Net, \zeta_\Sys)\) is \emph{weakly final} among the \(F\)-coalgebras that satisfy \(E\), in that every \((X, \beta) \in \Coalg(F, \Sys)\) has a coalgebra homomorphism into \((\lim \Net, \zeta_\Sys)\).                                                    
    We need to show that it is indeed \emph{final}, in the sense that this coalgebra homomorphism is unique. 
    Given \emph{any} coalgebra homomorphism \(h \colon (X, \beta) \to (\lim \Net, \zeta_{\Sys})\), it suffices to see that \(h\) is a cone homomorphism \((X, \lambda) \to (\lim \Net, c^\bullet)\), where \(\lambda\) was defined in \eqref{eq:lambda}, because the latter is a limit cone.       
    For this, we start by observing the identity \(\zeta_\Sys^w \circ h = \WFunc w(h) \circ \beta^w\) for all \(w \in \Sys_F^*\), by induction on \(w\) (similar to \cref{lem:always a homom}) and by virtue of \(h\) being a coalgebra homomorphism.                                              
    Therefore, for \(w \in \Sys_F^*\),
    \begin{align*}            
        c_w^\bullet \circ h   
        &= \WFunc w(c_\bullet^\bullet) \circ \zeta_\Sys^w \circ h                                            
            \tag{\(\zeta_\Sys^w \colon (\lim \Net, c^\bullet) \to (\WFunc w(\lim \Net), \WFunc w(c^\bullet))\)}                             
        \\                    
        &= \WFunc w(\bang_{\lim \Net}) \circ \zeta_\Sys^w \circ h                                            
            \tag{\(\bang_{\lim \Net} = c_\bullet^\bullet \colon \lim(\WFunc \bullet \circ \Net) \to \Net \bullet = 1\)}                     
        \\                    
        &= \WFunc w(\bang_{\lim \Net}) \circ \WFunc w(h) \circ \beta^w                                       
            \tag{observation above}   
        \\                    
        &= \WFunc w(\bang_{\lim \Net} \circ h) \circ \beta^w                                                 
            \tag{\(\WFunc{w}\) a functor}        
        \\                    
        &= \WFunc w(\bang_X) \circ \beta^w \\    
        &= \lambda_w \tag{def.~of \(\lambda\)}   
    \end{align*}              
    as desired.               
    It follows that \(h = \beh_\beta\), and finally that \((\lim \Net, \zeta_{\Sys})\) is the final \(F\)-coalgebra that satisfies \(\Sys\).
\end{proof}

For \(\Sys = \emptyset\), \cref{thm:generalized Adamek-Barr construction} is precisely the terminal sequence construction.
Note that this means that \(\Pow\) is not pitched-continuous, as it is not \(\omega^{\op}\)-continuous.
However, polynomial functors, like \(\Delta_B \times \Id^A\) on \(\Set\), are pitched-continuous, and provide us with a number of interesting examples of when the terminal net construction produces a relatively final coalgebra.
This follows directly from the facts that every pitched diagram is \emph{connected} and that polynomial functors on \(\Set\) preserve connected limits (see~\cite{Pareé_1990}).

\begin{restatable}{lem}{polynomialpitched}
    \label{lem:polynomial pitched}
    Every polynomial endofunctor on \(\Set\) is pitched-continuous.
\end{restatable}

Since every equational path constraint \(E \subseteq F^n\) is an equalizer, if \(F\) is pitched-continuous, then so is \(E\).
Combining this with \cref{thm:acc set eq} tells us that for any polynomial functor \(F \colon \Set \to \Set\), every singular path constraint is equational.
We immediately obtain the following from \cref{thm:generalized Adamek-Barr construction,lem:polynomial pitched,thm:acc set eq}.

\begin{cor}
	\label{cor:set based terminal net}
	Let \(F \colon \Set \to \Set\) be a polynomial functor and let \(\Sys\) be a system of singular path constraints for \(F\)-coalgebras. 
	Then the \(F\)-coalgebra \((Z_\Sys, \zeta_\Sys)\) obtained in \cref{thm:generalized Adamek-Barr construction} is a final \(F\)-coalgebra relative to \(\Sys\).
\end{cor}

\section{Future and Related Work}
\label{sec:discussion}

Let us start by addressing areas where further effort is needed.
We are currently missing structural characterizations of covarieties that are equational over specific choices in \(H\). 
Finding useful characterizations appears to be difficult, even for transition systems and frame conditions; see \cref{eg:transitivity again,rem:frame}. 
Also, we have touted ease-of-specification as a strength of path constraints, but only explicitly provided one illustration in \cref{eg:transitivity again}.
Syntactic specifications of path constraints are also implicit in \cref{thm:monoid presentations,eg:independence}, but a more comprehensive study is needed.
We would also like a more general version of \cref{thm:generalized Adamek-Barr construction} for nonsingular path constraints, as well as general versions of \cref{thm:converse to covariety,thm:acc set eq} for locally presentable categories~\cite{AdamekR94}.
\smallskip

As for related work, there are several branches of the literature that are clearly relevant and contain similar results to ours. 

\paragraph{Coalgebraic modal logic} 
The literature most easily compared with path constraints comes from \emph{coalgebraic modal logic}~\cite{Moss1999}.
Classically, the modal logic of possibility and necessity obtains a semantics from Kripke models, which are transition systems (\cref{eg:examples of coalgebras}\eqref{eg:it:transition}) with a labelling by propositional formulas.
This is the traditional setting of Kripke frames as studied in~\cref{eg:transitivity again}.
Coalgebraic modal logic generalizes the classic modal logic of possibility and necessity from Kripke frames and models to general coalgebras.
Developments in coalgebraic modal logic since its introduction are a bit overwhelming in number, so we will just point the reader to~\cite{CirsteaKPSV11,KupkeP11,KupkeR21,SchroederExpress} for references to more recent results.

For us, the point of similarity between path constraints and formulas in coalgebraic modal logics begins with Kurz~\cite{Kurz98}, who proposes that modal formulas be taken to be synonymous with coequations, and uses Moss's modal logic~\cite{Moss1999} as an example where every coequation is equivalent to a formula.
We would be very interested to see a comprehensive translation between modal formulas and equational path constraints over an interesting \(H\) by some concrete means.

An idea from the coalgebraic modal logic literature that closely resembles ours is the notion of \emph{path} used by Goldblatt~\cite{Goldblatt02} (based on work of Jacobs~\cite{Jacobs00}) in his equational logic for polynomial coalgebras.
Although they were not phrased as such, paths through polynomial coalgebras are special cases of natural transformations \(F \Rightarrow H\) for polynomial \(F\) and \(H\) that are compositions or projections, reverse inclusions, and evaluations. 
Goldblatt even uses them in a similar way as we use left and right transformations, as parts of equations.
Goldblatt's logic is less expressive than equational path constraints in general, because it cannot express properties like transitivity (\cref{eg:transitivity again}).
A more precise relationship between Goldblatt's equational logic for polynomial coalgebras and our equational path constraints would be interesting to discover.

\paragraph{Duality of equations and coequations}
There is a significant line of work in which equations and coequations are related by exploiting the duality of algebra and coalgebra~\cite{Ballester-BolinchesCRS2015,SalamancaBR16,GherkeGP08}.
Intuitively, if \(\Cat\) and \(\mathbf{D}\) are dual categories, then varieties of algebras for an endofunctor on \(\mathbf D\) are dual to covarieties of coalgebras for the corresponding dual endofunctor on \(\Cat\), allowing for a two-way translation of axiomatization and coaxiomatization techniques between these categories.

In~\cite{Ballester-BolinchesCRS2015}, in particular, the authors relate equations and coequations for classic deterministic automata, which are Moore automata with two outputs (\(B = 2\)).
In their setting, coequations are subsets of a final automaton, i.e., coequations in \(1\) colour.
They show that congruences on the monoid of words \(A^*\) are in dual correspondence with these coequations by exhibiting a relatively final (in their terminology, \emph{cofree}) coalgebra satisfying the congruence~\cite[Theorem 22]{Ballester-BolinchesCRS2015}.
This work is generalized to the more abstract setting of dual equivalences in~\cite{SalamancaBR16}.
In both papers (\cite{SalamancaBR16}~and~\cite{Ballester-BolinchesCRS2015}), concrete relatively final automata are constructed using a method similar to ours (compare \cref{eg:grids,eg:biinfinite streams} with Examples 26, 27, and 40 in~\cite{Ballester-BolinchesCRS2015} and~\cite[Section 7.2]{SalamancaBR16}).
\cref{thm:monoid presentations} generalizes these constructions to arbitrary sets of output symbols, and \cref{thm:generalized Adamek-Barr construction} to arbitrary polynomial functors.

\paragraph{Hidden algebra}
The coalgebraic picture of hidden algebra~\cite{Cirstea97,Worrell98} is also worth mentioning, since it deals with behavioural properties of coalgebras specified as equations.
To quote Worrell~\cite{Worrell98}, hidden algebras can be seen as deterministic coalgebras.
In our view, this implies a well-defined notion of \emph{path}.
Paths in a hidden algebra can be specified by terms from an algebra, and one can write equations between terms to describe behavioural properties. 
These appear to be similar to equational path constraints over \(\Id\), but the precise connection is still unclear at the time of writing. 

\paragraph{Logic of coequations}
Let us also mention the work of Ad\'{a}mek~\cite{Adamek05} and Ad\'{a}mek and Schwencke~\cite{Schwencke10}, where the authors introduce and develop a \emph{logic of coequations} that differs from coalgebraic modal logic in that it directly deals with behaviours. 
The logic of coequations allows for deductive reasoning about avoidance of behaviour patterns in coalgebras.
In~\cite{Adamek05}, Ad\'{a}mek shows that a covariety has a chromatic number \(\le \kappa\) if and only if it is closed under a strengthened notion of bisimilarity called \emph{\(\kappa\)-bisimilarity}, improving an earlier result of Gumm and Schr\"oder~\cite{GummS01}.
This turns out to be useful for giving lower bounds on chromatic numbers: for example, paired with \cref{thm:coax}, one can show that the commutative Moore automata have a chromatic number of exactly \(2\).

\paragraph{Coalgebra for differential equations}
Finally, recall that our initial motivation came from a paper of Boreale~\cite{Boreale19}.
This line of work originates in the correspondence between differential equations and coinductive specifications of streams introduced by Pavlovic and Escardo~\cite{PavlovicE98}.
The latter paper led to Rutten's \emph{stream calculus}~\cite{Rutten01}, which provides a syntax for algebraically manipulating streams from a field.

By generalizing the universal property of streams to the universal property of formal power series in commuting variables (\cref{eg:grids}), one can rather straightforwardly obtain a higher-dimensional version of Rutten's stream calculus.
\cref{thm:monoid presentations} appears to take this a step further to arbitrary monoid-indexed sets, and it is not difficult to imagine what a corresponding calculus of monoid-indexed sets from a field might look like.
Again, we leave working out the details for future work.

\bibliographystyle{alphaurl}
\bibliography{refs}

\end{document}